\documentclass{article}
\usepackage{subfiles}
\usepackage{xr-hyper}
\usepackage[usenames,dvipsnames]{xcolor}
\usepackage{amsfonts}
\usepackage{amsmath,amsthm,amssymb,dsfont}
\usepackage{enumerate}
\usepackage[english]{babel}
\usepackage{graphicx}	
\usepackage[caption=false]{subfig}
\usepackage[margin=1in]{geometry}
\usepackage{url}
\usepackage{todonotes}
\usepackage{bbm}
\usepackage{booktabs}
\usepackage{hyperref} 

\usepackage{amsmath,amssymb,amsthm,dsfont,mathrsfs}
\usepackage{makecell}
\usepackage{array,multirow,graphicx}
\usepackage{setspace}
\usetikzlibrary{arrows.meta}
\newcommand{\E}{\mathbb{E}}
\newcommand{\rEm}{\mathrm{E}_{\max}}
\newcommand{\Var}{\mathrm{Var}}
\newcommand{\Cov}{\mathrm{Cov}}

\usepackage{algorithm}
\usepackage[noend]{algorithmic}
\usepackage{enumitem}
\usepackage{tikz}
\usetikzlibrary{chains}
\usetikzlibrary{fit}
\usepackage{pgflibraryarrows}		
\usepackage{pgflibrarysnakes}		

\usepackage{makecell}

\usepackage{epsfig}
\usetikzlibrary{shapes.symbols,patterns} 
\usepackage{pgfplots}
\pgfplotsset{compat=newest}

\usepackage{mathrsfs}

\usepackage{hyperref}
\hypersetup{colorlinks=true,citecolor=blue,linkcolor=blue,filecolor=blue,urlcolor=blue,breaklinks=true}

\usepackage{nicefrac}
\usepackage{mathtools}
\usepackage{mathtools}
\usepackage[utf8]{inputenc}
\usepackage[english]{babel}
\usepackage{amssymb,amsmath,amsthm}

\newtheorem{theorem}{Theorem}
\newtheorem{lemma}{Lemma}

\theoremstyle{definition}

\theoremstyle{definition}
\newtheorem{definition}{Definition}

\newcommand{\Ac}{\mathcal{A}}
\newcommand{\Bc}{\mathcal{B}}

\newcommand{\Dc}{\mathcal{D}}
\newcommand{\Ec}{\mathcal{E}}

\newcommand{\Gc}{\mathcal{G}}

\newcommand{\Tc}{\mathcal{T}}

\newcommand{\EE}{\mathbb{E}}
\newcommand{\PP}{\mathbb{P}}

\newcommand{\ER}{\mathrm{ER}}
\newcommand{\rE}{\mathrm{E}}
\newcommand{\PV}{\nu}
\newcommand{\Hc}{\mathcal{H}}

\newcommand{\mb}{\overline{m}}
\DeclarePairedDelimiter\ceil{\lceil}{\rceil}

\newcommand{\bigO}{\mathcal{O}}

\newcommand{\barD}{\overline{\mathcal{D}}}

\usepackage{amsthm}

\usepackage[affil-it]{authblk}

\author[1]{Huy Pham}
\author[1]{Hoang Ta}
\affil[1]{\small Department of Computer Science, Hanoi University of Science and Technology, Vietnam}

\begin{document}

\title{{\LARGE A Fast Hierarchical Splitting Approach for Non-Adaptive Learning of Random Hypergraphs}}

\maketitle


\begin{abstract}     
This work focuses on the problem of learning an unknown $3$-uniform hypergraph using hyperedge-detecting queries. Our goal is to design a querying strategy that recovers the hyperedge set using as few queries as possible. We restrict our attention to random hypergraphs under the Erd\H{o}s--R\'enyi (ER) model, in which each potential hyperedge appears independently with probability $q = \Theta(n^{-3(1-\theta)})$ for $\theta \in (0;1)$. Prior work [Austhof-Reyzin-Tani, ISIT 2025] presents a testing-decoding scheme that uses $\bigO(\mb \log n)$ tests but requires a decoding time of $\Omega(n^3)$, where $\mb = q\binom{n}{3}$ denotes the expected number of hyperedges. 

In this work, we extend the binary splitting framework and adapt it to the $3$-uniform hypergraph setting. We obtain a testing-decoding scheme that recovers the hyperedge set with high probability using $\bigO(\mb \log n)$ tests and achieves decoding time $\bigO(\mb^{5/3}\log n)$ for the case $\theta > \dfrac{2}{3}$ and $\bigO(\mb^{5/3}\log^2{\mb}\log n)$ for the case $\theta \leq \dfrac{2}{3}$, the resulting runtime is subcubic in $n$ when $\theta < \dfrac{3}{5}$ and provides a different test–decoding tradeoff from existing schemes.
\end{abstract}


\section{Introduction}
\label{sec:intro}
In the graph learning problem with edge-detecting queries, we are given a graph \( G = (V, E) \) in which the edge set \( E \) is unknown. Our goal is to recover \( E \) by issuing queries on subsets of the vertex set. Each query asks whether the chosen subset contains at least one edge entirely within it; such queries are called \emph{edge-detecting queries}. The main challenge is to design querying strategies that accurately recover the underlying edge structure while using as few queries as possible. This problem has been widely studied due to its applicability in uncovering hidden relationships within a set of objects, where each object is represented as a vertex.

Most existing formulations focus on pairwise relationships, modeling the structure as a graph in which edges connect exactly two vertices. However, this abstraction can be inadequate in settings where interactions inherently involve multiple entities. To better capture such phenomena, it is natural to move beyond graphs and consider hypergraphs, where edges are replaced by hyperedges that span subsets of vertices. Among these, $3$-uniform hypergraphs—where each hyperedge contains exactly three vertices—represent the simplest nontrivial extension beyond pairwise models. In this work, we are given a $3$-uniform hypergraph and aim to recover its hyperedge set using as few hyperedge-detecting queries as possible. This problem has a wide range of applications in areas such as chemical reactions, molecular biology, and genome sequencing.

Dana Angluin and Jiang Chen have shown in \cite{ANGLUIN2008546} that at least $\Omega((2m/r)^{r/2})$ edge-detecting queries are required to learn a general hypergraph with dimension $r$ and $m$ edges. In this paper, we focus on recovering the hyperedge set of a $3$-uniform hypergraph under the Erd\H{o}s--R\'enyi (ER) model, in which each potential hyperedge appears independently with probability $q$. In particular, we generalize and extend the binary splitting framework for non-adaptive group testing~\cite{price2020fast} and non-adaptive graph learning~\cite{ta2025fastbinarysplittingapproach}. We later show that our algorithm requires \( O(\mb \log n) \) tests and achieves a decoding time of \( O(\mb^{5/3} \log^2 \mb \log n) \).

Although the graph and hypergraph settings may appear similar at first glance, extending the binary splitting framework from graphs to $3$-uniform hypergraphs introduces substantial new difficulties.

\begin{itemize}
    \item We introduce a new characterization of \emph{typical hypergraphs}. Unlike the graph setting, where it suffices to control defective blocks and block degrees, the hypergraph setting requires tracking several distinct classes of block interactions arising from hyperedges occupying one, two, or three blocks simultaneously. This richer structural characterization underpins our probabilistic analysis.

    \item We develop a new probabilistic analysis of the recursive refinement process. Although a single candidate triple may generate substantially more descendants than in the graph setting, we show that the number of surviving candidate triples remains bounded with high probability across all levels of the recursion.

    \item We extend the binary splitting framework from graphs to $3$-uniform hypergraphs by introducing a refinement procedure and accompanying analysis that handle all possible distributions of a hyperedge across recursive block partitions, while preserving the order-optimal $O(m\log n)$ test complexity and achieving efficient decoding.
\end{itemize}

\subsection{Related Work}
A hypergraph is a generalization of a graph in which a hyperedge can contain more than two vertices, rather than just a pair of vertices as in ordinary graphs. Hypergraphs have been widely studied because they allow for a richer and more flexible representation of relationships among entities. As an illustrative example, consider a chemical reaction network with $n$ distinct compounds, where a pair or a mixture of compounds may interact. In such settings, ordinary graphs fail to adequately capture the underlying relationships, since a chemical reaction may involve more than two compounds simultaneously. Hypergraphs, however, naturally model these multi-way interactions and provide a more suitable framework for representing such systems.

These interactions are usually referred to as higher-order interactions (HOIs), and they arise in a wide range of fields, such as physical systems \cite{Battiston2021}, social contagion models \cite{Social_Contagion}, and brain functions \cite{Expert2022}. Hypergraph neural networks have recently attracted significant attention in deep learning and machine learning due to their ability to model such complex higher-order relationships. For instance, \cite{Feng_You_Zhang_Ji_Gao_2019} proposes an HGNN framework for representation learning, demonstrating superior performance over state-of-the-art methods, particularly in handling multimodal data. Hypergraphs have also been applied in various problems such as brain classification \cite{Brain_Classification}, group recommendation \cite{Group_Recommendation}, and clinical prediction \cite{Clinical_Predictions}.

The graph learning problem with edge-detecting queries has been extensively studied (see \cite{Grebinski2000}, \cite{Beigel_2002}, \cite{alon2005learning}, \cite{bouvel2005combinatorial}). Angluin and Chen then generalized this learning problem to the hypergraph setting in \cite{angluin_chen}. They showed that an $r$-uniform hypergraph with $m$ edges and $n$ vertices can be learned with high probability using $\bigO(2^{4r} m \cdot \mathrm{poly}(r, \log n))$ queries. Moreover, these queries can be organized into $\bigO\bigl(\min(2^r (\log m + r)^2, (\log m + r)^3)\bigr)$ rounds. They further extended their results to almost uniform hypergraphs of dimension $r$, showing that such hypergraphs can be learned with $\bigO\bigl(2^{\bigO((1+\Delta/2)r)} \cdot m^{1+\Delta/2} \cdot \mathrm{poly}(\log n)\bigr)$ queries with high probability, where $\Delta$ denotes the difference between the maximum and minimum edge sizes. In a related line of work, Dyachkov et al.~\cite{dyachkov} study a restricted family $\mathcal{F}(t, s, \ell)$ of localized hypergraphs, where $|V| = t$, $|E| \le s$ with $s \ll t$, and each edge has size at most $\ell$, with $\ell \ll t$. The objective is to identify all edges of an unknown hypergraph $H_{un} \in \mathcal{F}(t, s, \ell)$ using the minimum number of tests. They propose an adaptive algorithm that achieves the information-theoretic lower bound, requiring at most $s\ell \log_2 t (1 + o(1))$ tests in the worst case. In contrast, the work of \cite{learninglowdegree} focuses on the class of hypergraphs with bounded degree. In particular, they prove that there exists a non-adaptive algorithm with $\bigO\bigl((2n)^{\rho \Delta + 1} \log^2 n\bigr)$ queries that learns any $\rho$-near-uniform hypergraph $H$ with maximum degree $\Delta \ge 2$ with high probability. For the case of $3$-uniform random hypergraphs under the Erd\H{o}s--R\'enyi model, \cite{austhof2025non} proposes a scheme that requires $\bigO(\mb\log^2 \mb + \mb\log \mb\log^2 n)$ tests and achieves the same order of decoding time. Although the number of tests can be reduced to $\bigO(\mb\log n)$, this improvement comes at the cost of a decoding time that is at least cubic in $n$. Recent work by Sheffield, Vassilevska Williams, and Xi \cite{sheffieldlimits}
studies exact recovery of hypergraphs under all-pairs detection oracles.
In the case of $3$-uniform hypergraphs with $n$ vertices and $m=|E|$ hyperedges, under pairwise oracle access, they obtain a recovery algorithm with complexity $\widetilde{\mathcal O}\!\left(n^{\frac{3}{4-c}} \, m^{\frac{3-c}{4-c}}\right),$ where $n^c$ denotes the cost of a single oracle query. In particular, when $c=2$, this yields a bound of $\widetilde{\mathcal O}(n^{3/2} m^{1/2})$.

The problem of learning a hypergraph can also be reformulated as learning a monotone DNF with at most $s$ terms, where each term contains at most $r$ variables. In particular, an $s$-term $r$-monotone DNF (MNDF) corresponds to a hypergraph learning problem with edge-detecting queries, where the hypergraph has $s$ edges and each edge contains at most $r$ vertices. For the case $s < r$, Hasan Abasi, Nader H. Bshouty, and Hanna Mazzawi~\cite{abasi_hasan} established a lower bound of $\bigO\bigl((r/s)^{s-1} + rs \log n\bigr)$ for learning $s$-term $r$-MNDFs. When $s \ge r$, Angluin and Chen~\cite{ANGLUIN2008546} showed a lower bound of $\Omega\bigl((2s/r)^{r/2} + rs \log n\bigr)$. For this regime, Abasi, Bshouty, and Mazzawi also proposed two nearly optimal algorithms: a deterministic algorithm requiring $(crs)^{r/2 + 1.5} + rs \log n$ membership queries for some constant $c$, and a randomized algorithm using $(c's)^{r/2 + 0.75} + rs \log n$ membership queries for some constant $c'$.

Moreover, the hypergraph learning problem can be interpreted as an instance of a complex group testing problem. In the classical group testing setting, given $n$ items with at most $d$ positives, a test on any subset returns positive if and only if the subset contains at least one positive item. A natural extension is the \emph{complex group testing} model, where a single positive item is not sufficient to yield a positive outcome. Instead, only specific combinations of items, called positive subsets, can trigger a positive result. More precisely, there exist at most $d$ such subsets, each of size at most $s$, and a test is positive if and only if it fully contains at least one of these subsets. This formulation closely aligns with the hypergraph learning problem: items correspond to vertices, positive subsets correspond to hyperedges, and the objective is to identify all hyperedges using as few tests as possible. The work of~\cite{CHIN201311} introduces a derandomized construction with two key properties: the number of tests is $\bigO\left(\frac{(d+s)^{d+s+1}}{d^d s^s} \cdot \log n\right)$, matching known bounds when $s = 1$, and each test contains exactly $k$ items for a prescribed parameter $k$.

\subsection{Technical overview and contributions}
Our work extends the technical framework of the binary splitting approach and adapts it to a non-adaptive setting for hypergraph learning using hyperedge-detecting queries under the Erd\H{o}s--R\'enyi model $\mathrm{ER}(n,q)$. Here, $n$ denotes the number of vertices and $q = \Theta\big(n^{-3(1-\theta)}\big)$ for $\theta \in (0,1)$ represents the edge probability. Let $\mb$ denote the expected number of edges in a random hypergraph $G \sim \mathrm{ER}(n,q)$, i.e., $\mb = q \binom{n}{3}$. In analogy with the graph learning setting, we treat edges as \emph{defective} triples of vertices.

We introduce a testing and decoding scheme that, with high probability, guarantees exact recovery of the hyperedge set of \(G\) using \(\mathcal{O}(\mb\log n)\) tests, while achieving decoding complexity \(\mathcal{O}(\mb^{5/3}\log^2{\mb}\log n)\) for $\theta \leq 2/3$ and \(\mathcal{O}(\mb^{5/3}\log n)\) for $\theta > 2/3$. This substantially improves upon the cubic-time decoding barrier of existing non-adaptive approaches while preserving near-linear test complexity. Our approach can be summarized as follows:
\begin{itemize}
    \item Building upon the binary splitting framework introduced in \cite{price2020fast} and \cite{ta2025fastbinarysplittingapproach}, we partition the vertex set into a ternary hierarchy of blocks across levels $\ell = \left\lceil\log_3{\mb^{1/3}}\right\rceil, \left\lceil\log_3{\mb^{1/3}}\right\rceil + 1, \dots, \log_3{n}$. At each level $\ell$, there are $3^{\ell}$ blocks denoted by $\{\mathcal{G}_1^{\ell}, \dots, \mathcal{G}_{3^{\ell}}^{\ell}\}$. For every level, we perform $\Theta(\mb^{2/3})$ independent repetitions. In each repetition, the blocks are assigned independently and uniformly at random to one of $\Theta(\mb^{1/3})$ tests, where assigning a group corresponds to including all of its vertices in the corresponding test. Since there are $\mathcal{O}(\log n)$ levels, the overall number of tests is $\mathcal{O}(\mb\log n)$.

    \item Given all test outcomes, the decoder runs deterministically in a coarse-to-fine manner across the hierarchy. It maintains a collection of possibly defective (PD) triples, discards any block triple contained in a negative test, and recursively refines the remaining triples into corresponding child triples at the next level. At the final level (i.e., singleton vertices), any triple that is not eliminated by a negative test is declared to be a hyperedge. In Theorem~\ref{thm:main}, we establish that, with high probability, this procedure successfully recovers all hyperedges of $G$ using $\bigO(\mb\log n)$ tests. Moreover, the decoding complexity is $\bigO\!\left(\mb^{5/3} \log n\right)$ when $\theta > 2/3$, and $\bigO\!\left(\mb^{5/3} \log^2 \mb \log n\right)$ when $\theta \leq 2/3$. Here, the $\mathcal{O}(\mb^{5/3})$ factor comes from the fact that we need to check $\Theta(\mb^{2/3})$ candidate tests per retained triple across levels.
\end{itemize}

Our results, as well as the existing results for the Erd\H{o}s--R\'enyi model, are summarized in Table \ref{tab:placeholder}.

\begin{table}[htbp]
    \centering
    \begin{tabular}{|c|c|c|}
    \hline
         \textbf{Reference} & \textbf{Number of tests} & \textbf{Decoding time}  \\
    \hline 
        GROTESQUE \cite{austhof2025non} & $\bigO(\mb\log^2{\mb} + \mb\log{\mb}\log^2{n})$ & $\bigO(\mb\log^2{\mb} + \mb\log{\mb}\log^2{n})$ \\
    \hline
        COMP \cite{austhof2025non} & $\mathcal{O}(\mb\log n)$ & $\Omega(n^3)$  \\
    \hline
         \textbf{Our Theorem \ref{thm:main}} & $\mathcal{O}(\mb\log n)$ & 
         $\bigO\!\left(\mb^{5/3} \log^2 \mb \log n\right)$  \\
    \hline
    \end{tabular}
    \caption{Overview of existing non-adaptive schemes for learning Erd\H{o}s--R\'enyi 3-uniform hypergraphs. Here, $n$ is the number of vertices and $\mb = q\binom{n}{3}$ is the expected number of edges.}
    \label{tab:placeholder}
\end{table}

While several components of our framework admit a formal extension to $k$-uniform hypergraphs, a major obstacle arises when $k$ is treated as a variable parameter, since the recursive refinement step generates exponentially many candidate tuples in $k$. For this reason, the present paper focuses on the $3$-uniform case.

\section{Problem Setup}
\label{sec:setup}
For clarity, we begin by considering $3$-uniform hypergraphs. Our objective is to learn the hyperedge structure of an unknown hypergraph $H = (V,E)$, where the vertex set is $V = \{1,2,\dots,n\}$ and the hyperedge set $E \subseteq \binom{V}{3}$ is random.

We adopt the Erd\H{o}s--R\'enyi (ER) hypergraph model~\cite{austhof2025non}, in which each candidate hyperedge $(u,v,w) \in \binom{V}{3}$ is included independently with probability $q = q(n)$, so that $H \sim \ER(n,q)$. After being generated, the hypergraph remains fixed but unknown, and the learner does not have direct access to the edge set $E$.

Information about $H$ is acquired through \emph{hyperedge-detecting queries}. Given a subset $S \subseteq V$, a query returns a single bit indicating whether $S$ contains at least one hyperedge from $E$. The goal is to construct a \emph{non-adaptive} set of such queries—fixed in advance—along with a decoding procedure that accurately reconstructs $E$. A more detailed description of the model follows.

Each query can be represented by a binary vector $X \in \{0,1\}^n$, where $X_i = 1$ signifies that vertex $i$ is included in the query. The corresponding outcome is $Y \;=\; \displaystyle\bigvee_{(u,v,w)\in E} \bigl( X_u \land X_v \land X_w \bigr),$
meaning that the output equals one if there exists at least one hyperedge in $E$ fully contained within the selected vertices, and zero otherwise.

In the \emph{non-adaptive} setting, the query vectors $X^{(1)},\dots,X^{(t)}$ are predetermined and cannot depend on previous outcomes, making the recovery problem more challenging compared to the adaptive setting. Based on the observed outcomes $\{Y^{(i)}\}_{i=1}^t$, a decoder produces an estimate $\widehat{H} = (V,\widehat{E})$. The performance is quantified by the error probability $P_e \;\coloneqq\; \PP\,\!\big[\,\widehat{E}\neq E\,\big],$ where the probability is taken over both the randomness of the Erd\H{o}s--R\'enyi hypergraph and the query design. The goal is to design a scheme such that $P_e \to 0$ as $n \to \infty$.
\medskip
\subsection{Mathematical and Computational Assumptions}
We focus on sparse $3$-uniform hypergraphs whose sparsity is parameterized by a constant
$\theta \in (0,1)$ via $q \;=\; \Theta\!\bigl(n^{-3(1-\theta)}\bigr),$ so that the expected number of hyperedges satisfies $\mb \;\coloneqq\; \binom{n}{3}\,q \;=\; \Theta\!\bigl(n^{3\theta}\bigr).$

Equivalently, as $\theta$ ranges over $(0,1)$, we have $n^{-3} \ll q \ll 1
   \quad\text{and}\quad
   1 \ll \mb \ll n^3,$
where $f(n)\ll g(n)$ abbreviates $f(n)=o\bigl(g(n)\bigr)$.

Throughout the paper, we work in the unit-cost word-RAM model. With $n$ vertices and $T$ tests reading any integer in $\{1,\dots,n\}$, performing basic arithmetic operations on such integers, and retrieving any test outcome (indexed by $\{1,\dots,T\}$) each takes $\bigO(1)$ time.

Without loss of generality, we assume that $n$ is a power of three.
When this is not the case, we may augment the hypergraph by adding at most
$3^{\lceil \log_3 n \rceil} - n$ dummy vertices so that the total number of
vertices equals the next power of three. Since these additional vertices are
isolated and do not participate in any hyperedges, they do not affect the
query outcomes or the recovery procedure. Consequently, all subsequent
results remain valid under this assumption.
\subsection{Level Hypergraphs}
\label{sec:level_hypergraphs}

For each $\ell \in \{\lceil \log_3 \mb^{1/3} \rceil,\dots,\log_3 n\}$, let
$g \coloneqq 3^{\ell}$, and let
$\{\mathcal{H}_1,\dots,\mathcal{H}_g\}$ be a balanced partition of $V$
into $g$ blocks, each of size $n/g$, where
\[
  \mathcal{H}_i \coloneqq \left\{ (i-1)\dfrac{n}{g}+1,\ \dots,\ i\dfrac{n}{g} \right\},
  \qquad i \in [g].
\]
A block $\mathcal{H}_i$ is called \emph{defective} if it contains at least one
hyperedge of $H$. For $t \in \{2,3\}$, we call
$(\mathcal{H}_{i_1},\dots,\mathcal{H}_{i_t})$ a \emph{defective}
if the union $\mathcal{H}_{i_1} \cup \dots \cup \mathcal{H}_{i_t}$ contains
at least one hyperedge. For $3$-uniform hypergraphs, we introduce the following parameters.
\begin{itemize}
    \item $(\Hc_g) = (V_g,E_g):$ The induced hypergraph at level $\ell$, where $g = 3^{\ell}, V_g = [g]$ and a hyperedge $(i,j,k)$ is present if and only if $(\Hc_i,\Hc_j,\Hc_k)$ is defective. We refer to $\Hc_g$ as the level-$\ell$ block hypergraph. 
    \item For a non-defective single block $\mathcal{H}_i$:
  \begin{itemize}
    \item $d_{g}^{(1,1)}(\mathcal{H}_i)$ is the number of indices
    $j \in [g]\setminus\{i\}$ such that
    $(\mathcal{H}_i,\mathcal{H}_j)$ is a defective $2$-triple but $\Hc_j$ is not defective. 

    \item $d_{g}^{(1,2)}(\mathcal{H}_i)$ is the number of pairs
    $(j_1,j_2)$ with $j_1,j_2 \in [g]\setminus\{i\}$ and $j_1 \neq j_2$
    such that $(\mathcal{H}_i,\mathcal{H}_{j_1},\mathcal{H}_{j_2})$ is a
    defective $3$-triple, but all the pairs $(\Hc_{i}, \Hc_{j_1}), (\Hc_{i}, \Hc_{j_2})$ and $(\Hc_{j_1}, \Hc_{j_2})$ are not defective.
  \end{itemize}

  \item For a non-defective pair $(\mathcal{H}_{i_1},\mathcal{H}_{i_2})$:
  \begin{itemize}
    \item $d_{g}^{(2,1)}(\mathcal{H}_{i_1},\mathcal{H}_{i_2})$ is the number
    of indices $j \in [g]\setminus\{i_1,i_2\}$ such that
    $(\mathcal{H}_{i_1},\mathcal{H}_{i_2},\mathcal{H}_j)$ is a defective
    $3$-triple, but all the pairs $(\Hc_{i_1}, \Hc_{j}), (\Hc_{i_2}, \Hc_{j})$ are not defective.
  \end{itemize}

  \item We then define the global parameters 
  \begin{align*}
d_{g}^{(1,1)} &:= \max_{\mathcal{H}_i \text{ non-defective}} d_{g}^{(1,1)}(\mathcal{H}_i),\\
d_{g}^{(1,2)} &:= \max_{\mathcal{H}_i \text{ non-defective}} d_{g}^{(1,2)}(\mathcal{H}_i),\\
d_{g}^{(2,1)} &\coloneqq \max_{(\mathcal{H}_{i_1},\mathcal{H}_{i_2}) \text{ non-defective}}
d_{g}^{(2,1)}(\mathcal{H}_{i_1},\mathcal{H}_{i_2}).
\end{align*}

  \item Finally, let $\nu_g^1$ denote the number of defective blocks among $\mathcal{H}_1,\dots,\mathcal{H}_g$, and let $\nu_g^2$ denote the number of defective pairs of blocks in which neither block is defective.
\end{itemize}

\subsection{Typical Hypergraphs}
\label{sec:typical_graphs}

To analyze the performance of the proposed decoding algorithms, we restrict attention to a class of hypergraphs that exhibit regular behavior. Specifically, we consider a high-probability subset of Erd\H{o}s--R\'enyi hypergraphs in which the total number of hyperedges, as well as various structural properties across multiple levels of partitioning, are tightly concentrated around their expected values. We refer to this collection as the \emph{typical set of hypergraphs}, and it forms the basis of our subsequent analysis.

\begin{definition}
\label{eq:definition-typical-set}
Let $(\epsilon_n)_{n \in \mathbb{N}}$ be a sequence with $\epsilon_n \to 0$. We define the $\epsilon_n-$typical set of hypergraphs $T(\epsilon_n)$ as the collection of hypergraphs $G$ satisfying the following:
\begin{enumerate}[label=(\roman*)]
  \item
  The number of hyperedges is close to its expectation, i.e.,
  \[
    (1-\epsilon_n)\,\mb \leq M \leq (1+\epsilon_n)\,\mb,
    \qquad \text{where } M = |E| \text{ and } \mb = \E[|E|]\, .
  \]

  \item
At every level $\ell \in \left\{\left\lceil \log_3 \overline{m}^{1/3} \right\rceil,\dots,\log_3 n\right\}$
with $g = 3^\ell$, the quantities $|E_g|,\nu^1_g, \nu^2_g$, and
$d_g^{(1,1)}, d_g^{(1,2)}, d_g^{(2,1)}$ satisfy
\begin{equation}
|E_g| \le E_{\max} \coloneqq \begin{cases}
    12\mb & \theta > \dfrac{2}{3}\\[2ex]
    2\mb\log^2{\mb} & \theta \leq \dfrac{2}{3}
\end{cases}
\label{eq:eg}
\end{equation}
\begin{equation}
\nu_g^1 \le \nu_{\max}^1 \coloneqq \begin{cases}
    \dfrac{3\mb}{g^2} &  \theta > \dfrac{2}{3} \\[2ex]
    3\mb^{1/3} & \theta \leq \dfrac{2}{3}
\end{cases} 
\label{eq:v1}
\end{equation}

\begin{equation}
\nu_g^2 \le \nu_{\max}^2 \coloneqq
\begin{cases}
    \dfrac{9\mb}{g} &  \theta > \dfrac{2}{3} \\[2ex]
    9\mb^{2/3} & \theta \leq \dfrac{2}{3}
\end{cases}
\label{eq:v2}
\end{equation}

and
\begin{equation}
d_g^{(1,1)} \le d^{(1,1)}_{\max} \coloneqq 18\mb^{1/3},
\label{eq:d11}
\end{equation}
\begin{equation}
d_g^{(1,2)} \le d_{\max}^{(1,2)} \coloneqq 9\mb^{2/3},
\label{eq:d12}
\end{equation}
\begin{equation}
d_g^{(2,1)} \le d_{\max}^{(2,1)} \coloneqq 18\mb^{1/3}.
\label{eq:d21}
\end{equation}
\end{enumerate}
\end{definition}

Condition (i) ensures concentration of the global hyperedge count, while condition (ii) guarantees regularity across different partition scales: the numbers of induced hyperedges and defective blocks, as well as the maximum degree among non-defective nodes, are all bounded in a manner consistent with the sparsity regime. Note that since $g \ge \mb^{1/3}$, we have $\mb/g \le \mb^{2/3}$ and $\mb/g^{2} \le \mb^{1/3}$.

\begin{lemma} \label{lem:typical}
Fix $\theta \in (0,1)$, and let $H \sim \ER(n, q)$ be a random $3$-uniform
hypergraph with $q = \Theta\!\bigl(n^{-3(1-\theta)}\bigr)$. Then there exists
a sequence $(\epsilon_n)_{n \in \mathbb{N}}$ with $\epsilon_n \to 0$ such that
\[
  \PP\big[ H \in \Tc(\epsilon_n) \big] \to 1 \quad \text{as } n \to \infty.
\]
\end{lemma}
The proof can be found in Appendix~\ref{sec:appendix_lemm_typical}.

\section{Hierarchical Splitting Approach for Hypergraph Learning}
\label{sec:binary_spliting}

We now introduce our algorithm for non-adaptive hypergraph recovery. The approach builds on techniques from combinatorial group testing and hierarchical binary splitting~\cite{price2020fast,cheraghchi2020combinatorial}. It comprises two main components: the design of a collection of tests and a decoding procedure that reconstructs the hyperedge set from the observed outcomes.

The central idea is to recursively partition the vertex set into groups and test triples of these groups to infer the presence of hyperedges. At each level, the vertex set is refined into smaller subsets, and randomized test assignments are used to detect interactions among them. Crucially, negative test outcomes enable the algorithm to efficiently eliminate large sets of candidate hyperedge triples, thereby narrowing down the search space.

\subsection{Testing Procedure}
\label{sec:testing}

We construct a ternary tree of vertex blocks across levels \( \ell = \ceil{\log_3 \mb^{1/3}}, \ceil{\log_3 \mb^{1/3}} + 1, \dotsc, \log_3 n \). Each node in the tree corresponds to a block. At level \( \ell \), there are \( 3^{\ell} \) blocks \( \{\Hc^{(\ell)}_1, \dotsc, \Hc^{(\ell)}_{3^{\ell}}\} \), each containing \( n/3^{\ell} \) vertices. The testing is randomized: in each repetition, every block is assigned uniformly at random to one of \( C_1 \mb^{1/3} \) tests, and this procedure is repeated \( C_2 \mb^{2/3} \) times, for some constants \( C_1, C_2 > 0 \) to be chosen later. The number of repetitions is proportional to \( \mb^{2/3} \), which balances accuracy and test efficiency.

\begin{algorithm}[H]
	\caption{Testing Procedure}
	\label{alg:testing}
	\begin{algorithmic}[1]
		\REQUIRE Number of nodes \( n \), expected number of hyperedges \( \mb \), constants $C_1,C_2,C'$ with $C_1 \geq 155, C_2 = C_1^3, C' >4$
		\STATE Set \( \ell_{\min} \gets \ceil{\log_3 \mb^{1/3}} \)
		\FOR{each level \( \ell = \ell_{\min}, \dotsc, \log_3 n - 1 \)}
		\FOR{each iteration in \( \{1, \dotsc, C_2 \mb^{2/3}\} \)} 
        \STATE Initialize a sequence of $C_1 \overline{m}^{1/3}$ tests
		\FOR{each block \( j = 1, \dotsc, 3^{\ell} \)}
		\STATE Assign block \( \Hc^{(\ell)}_j \) to a randomly chosen test among the $C_1 \mb^{1/3}$ tests
		\ENDFOR
		\ENDFOR
		\ENDFOR
		\STATE At the final level $\ell = \log_3 n$, each block reduces to a singleton, i.e., $\Hc_{j}^{(\ell)} = \{j\}$. Repeat steps 3--6 for $C' \log n$ rounds, with each round consisting of $C_1 C_2 \mb$ tests.
	\end{algorithmic}
\end{algorithm}

\subsection{Decoding Procedure}
\label{sec:decoding}
The decoder is deterministic and processes the levels from coarse to fine. At each level $\ell$ we keep a set of \emph{possible defective} triples $\mathcal{PD}^{(\ell)}$, whose elements are block triples $(\Hc_i,\Hc_j,\Hc_k)$ that may still contain a true hyperedge at the current resolution.

If a triple $(\Hc_i,\Hc_j,\Hc_k)\in\mathcal{PD}^{(\ell)}$ is not eliminated, we refine it as follows.
Each block is partitioned into three children
\[
\Hc_i \to \{\Hc_i^{(L)},\Hc_i^{(M)},\Hc_i^{(R)}\}, \quad
\Hc_j \to \{\Hc_j^{(L)},\Hc_j^{(M)},\Hc_j^{(R)}\}, \quad
\Hc_k \to \{\Hc_k^{(L)},\Hc_k^{(M)},\Hc_k^{(R)}\}.
\]
Let
\[
\mathcal{C}_{ijk}
:= \{\Hc_i^{(L)},\Hc_i^{(M)},\Hc_i^{(R)},
      \Hc_j^{(L)},\Hc_j^{(M)},\Hc_j^{(R)},
      \Hc_k^{(L)},\Hc_k^{(M)},\Hc_k^{(R)}\}.
\]
We then add to $\mathcal{PD}^{(\ell+1)}$ all triples consisting of three distinct
elements chosen from $\mathcal{C}_{ijk}$, namely $\displaystyle\binom{|\mathcal{C}_{ijk}|}{3}.$ Since $|\mathcal{C}_{ijk}|=9$, this refinement produces $\displaystyle\binom{9}{3}=84$ child
triples. All of these triples can capture hyperedges potentially connecting three blocks, spanning two blocks, or lying entirely within one block.

In this way the procedure only removes triples or refines them into smaller ones. At the final (singleton) level, every surviving vertex triple not ruled out by any negative test is declared a hyperedge.

\begin{algorithm}[H]
	\caption{Decoding Procedure}
	\label{alg:decoding}
	\begin{algorithmic}[1]
		\REQUIRE Test outcomes \( \{Y^{(i)}\}_{i=1}^t \) from Algorithm~\ref{alg:testing}, number of nodes \( n \), expected number of hyperedges \( \mb \)

		\STATE Initialize candidate set
		\[
		\mathcal{PD}^{(\ell_{\min})}
		=
		\left\{
		(\Hc^{(\ell_{\min})}_i,\Hc^{(\ell_{\min})}_j,\Hc^{(\ell_{\min})}_k)
		:\;
		1 \le i < j < k \le 3^{\ell_{\min}}
		\right\}
		\]

		\FOR{each level \( \ell = \ell_{\min}, \dotsc, \log_3 n - 1 \)}
			\FOR{each triple \( (\Hc_i,\Hc_j,\Hc_k) \in \mathcal{PD}^{(\ell)} \)}
				\IF{no negative test contains \( \Hc_i \), \( \Hc_j \), and \( \Hc_k \)}
					\STATE Add all child triples generated from \( \Hc_i \), \( \Hc_j \), and \( \Hc_k \) to \( \mathcal{PD}^{(\ell+1)} \)
				\ENDIF
			\ENDFOR
		\ENDFOR

		\STATE Let \( \widehat{E} \) be the set of triples in \( \mathcal{PD}^{(\log_3 n)} \) that are not ruled out by any negative test at the final level

		\STATE Return hypergraph estimate \( \widehat{H} = (V, \widehat{E}) \)
	\end{algorithmic}
\end{algorithm}

%
%
\subsection{Algorithmic Guarantees}
\label{sec:algorithmic_guarantees}

In this section, we analyze the correctness and efficiency of our approach. We show that Algorithms~\ref{alg:testing} and~\ref{alg:decoding} achieve an order-optimal number of tests and, with high probability, recover the entire hyperedge set of the underlying Erd\H{o}s--R\'enyi hypergraph. We also establish the computational complexity of the decoding procedure. First, we will demonstrate that our algorithm will not output any false negative results.

\begin{lemma}[No False Negatives]
\label{lem:false_negative}
For every hyperedge $e\in E$ and every level $\ell$ of the recursive
decoding procedure, there exists a candidate block triple
$(\mathcal{H}_u^{(\ell)},\mathcal{H}_v^{(\ell)},\mathcal{H}_w^{(\ell)})$
such that
\[
e\subseteq
\mathcal{H}_u^{(\ell)}
\cup
\mathcal{H}_v^{(\ell)}
\cup
\mathcal{H}_w^{(\ell)}.
\]
Moreover, no such candidate triple can be eliminated by a negative test.
Consequently, every true hyperedge is preserved throughout the decoding
procedure, and hence $E\subseteq \widehat{E}.$
\end{lemma}

\begin{proof}
    We now proceed by induction on the level. At the initial level, the possibly defective set is defined as
    \begin{align*}
        \mathcal{PD}^{(\ell_{\min})}
        =
        \left\{
        (\Hc^{(\ell_{\min})}_i,\Hc^{(\ell_{\min})}_j,\Hc^{(\ell_{\min})}_k)
        : 1 \le i < j < k \le 3^{\ell_{\min}}
        \right\}.
    \end{align*}
    
    Fix an arbitrary hyperedge $e=\{x,y,z\}\in E$. At level $\ell_{\min}$, the vertices of $e$ may lie in one, two, or three distinct blocks. In all cases, there exist distinct indices $i,j,k$ such that
    $e\subseteq \Hc_i^{(\ell_{\min})}\cup \Hc_j^{(\ell_{\min})}\cup \Hc_k^{(\ell_{\min})}$.
    Hence, there exists a triple in $\mathcal{PD}^{(\ell_{\min})}$ whose union contains $e$. 
    
    Suppose that at some level $\ell$, there exists a triple
    $(\Hc_i^{(\ell)},\Hc_j^{(\ell)},\Hc_k^{(\ell)})\in\mathcal{PD}^{(\ell)}$
    such that $e\subseteq \Hc_i^{(\ell)}\cup\Hc_j^{(\ell)}\cup\Hc_k^{(\ell)}$.
    Since the triple contains the hyperedge $e$, there are only two possibilities. Either the three blocks are assigned to the same test, in which case the test contains $e$ and is therefore positive, or they are not assigned to the same test, in which case the triple is not examined in that round. In either case, the triple cannot be eliminated; therefore, it survives to the next level and is further refined.

    When the surviving triple is refined, each of its blocks is partitioned into child blocks at level $\ell+1$. The vertices $x,y,z$ therefore lie in one, two, or three of these child blocks. By the construction of $\mathcal{PD}^{(\ell+1)}$, these child blocks can be extended, if necessary, to a triple of distinct child blocks whose union contains $e$. Hence, there exists
    $(\Hc_{i'}^{(\ell+1)},\Hc_{j'}^{(\ell+1)},\Hc_{k'}^{(\ell+1)})\in\mathcal{PD}^{(\ell+1)}$
    such that $e\subseteq \Hc_{i'}^{(\ell+1)}\cup\Hc_{j'}^{(\ell+1)}\cup\Hc_{k'}^{(\ell+1)}$.

    Thus, by induction, at every level there exists a possibly defective block triple whose union contains $e$. At the final level, all blocks are singletons, so the surviving triple containing $e$ is precisely $(\{x\},\{y\},\{z\})$. Therefore, $e\in\widehat{E}$. Since $e\in E$ was arbitrary, we conclude that $E\subseteq\widehat{E}$.
\end{proof}

\begin{theorem}
\label{thm:main}
Fix \( \theta \in (0,1) \), and let \( H \sim \ER(n, q) \) with \( q = \Theta(n^{-3(1 - \theta)}) \). 
Let \( \mb = q \binom{n}{3} = \Theta(n^{3\theta}) \) denote the expected number of hyperedges. 
There exist constants \( C_1 \geq 155 \), \( C_2 = C_1^3 \), and \( C' > 4 \) such that, with \( \bigO(\mb \log n) \) tests, the following holds with probability at least $1 - o(1)$:
\begin{enumerate}[label=(\alph*)]
    \item If \( \theta > 2/3 \):
    \begin{itemize}
        \item The returned estimate \( \widehat{E} \) equals \( E \);
        \item The decoding time is \( \bigO \!\left(\mb^{5/3} \log n \right) \); 
    \end{itemize}
    
    \item If \( \theta \leq 2/3 \):
    \begin{itemize}
        \item The returned estimate \( \widehat{E} \) equals \( E \);
        \item The decoding time is \( \bigO \!\left(\mb^{5/3} \log^2 \mb \log n \right) \). 
    \end{itemize}
\end{enumerate}
\end{theorem}
From Lemma~\ref{lem:typical}, the Erdős--Rényi hypergraph belongs to the typical set (see Section~\ref{sec:typical_graphs}) with probability $1-o(1)$. Therefore, it suffices to fix an arbitrary hypergraph $H \in \Tc(\epsilon_n)$ and prove the theorem under this assumption. Throughout the remainder of the proof, the hypergraph $H$ is regarded as fixed, and all probabilities and expectations are taken only over the randomness of the testing procedure.

Before proving Theorem~\ref{thm:main}, we first estimate the size of $\mathcal{PD}^{(\ell)}$ for all levels $\ell$. 
For a given $\ell$, let $g = 3^{\ell}$. 
For $v \in [g]$, denote by $\Hc_v$ the $v$-th node, i.e., the $v$-th block in $H$. 
A node $\Hc_v$ is called \emph{non-defective} if $\Hc_v$ contains no hyperedge, and a triple $(\Hc_u,\Hc_v, \Hc_w)$ is called \emph{non-defective} if their union contains no hyperedges. 
For a single random test, let $Y$ denote the test outcome and $\mathcal{L}$ the set of node indices included in the test.

For distinct $u,v,w$, let $(\Hc_u, \Hc_v,\Hc_w)$ be a non-defective triple at level~$\ell$, and let $\Ec_{uvw}$ be the event that this triple is not identified at level~$\ell$ in Algorithm~\ref{alg:testing}. 
Let $\rE_{uvw}$ denote the corresponding indicator random variable, then the total number of unidentified non-defective triples at level $\ell$ is $\displaystyle\sum_{u,v,w} \rE_{uvw}$. For each node $u$, denote by $h(u) \in \{1,2,\dots,C_1\mb^{1/3}\}$ the index of the test containing $u$ in one iteration (out of $C_2\mb^{2/3}$ iterations). 
The dependence of these quantities on $\ell$ is left implicit. We write $\EE_{\ell}[\cdot]$ for conditional expectation given all test placements at earlier levels.


%
%
\begin{lemma}
\label{lem:expectation}
Let $H\in\Tc(\epsilon_n)$ be an arbitrary typical hypergraph. Conditioned on the $\ell$-th level having $|\mathcal{PD}^{(\ell)}| \leq 168 \rEm$, where $\rEm$ is given by Eq.~\eqref{eq:definition-typical-set}. Then, for any
$C_1\ge155$ and $C_2=C_1^3$
\[
\EE_{\ell}\!\left(\sum_{u,v,w}\rE_{uvw}\right)
\le
\frac{\rEm}{2},
\]
where the sum is taken over all non-defective triples in
$\mathcal{PD}^{(\ell)}$.
\end{lemma}

\begin{lemma}
\label{lem:variance}
Let $H\in\Tc(\epsilon_n)$ be an arbitrary typical hypergraph. Conditioned on the $\ell$-th level having $|\mathcal{PD}^{(\ell)}| \leq 168 \rEm$, where $\rEm$ is given by Eq.~\eqref{eq:definition-typical-set}. Then, for any
$C_1\ge155$ and $C_2=C_1^3$
\begin{align*}
   \Var_{\ell} \!\left(\sum_{u,v,w} \rE_{uvw} \right) \leq \bigO \!\left( \frac{\rEm^2}{\mb^{1/3}} \right),
\end{align*} 
where the sum is taken over all non-defective triples in $\mathcal{PD}^{(\ell)}$.
\end{lemma}

The proofs for Lemma \ref{lem:expectation} and Lemma \ref{lem:variance} can be found in the Appendix \ref{section:analysis_of_level_hypergraph}.


\begin{lemma}
\label{lem:bounds_size_PD}
Let $H\in\Tc(\epsilon_n)$ be an arbitrary typical hypergraph. Then, for any
$C_1\ge155$ and $C_2=C_1^3$, then
\begin{align*}
    \PP\left(
        \max_{\ell_{\min}\leq\ell\leq\log_3 n}
        |\mathcal{PD}^{(\ell)}|
        \leq168\rEm
    \right)
    \geq
    1-\bigO\left(\frac{\log_3 n}{\mb^{1/3}}\right).
\end{align*}
In other words, the probability that $|\mathcal{PD}^{(\ell)}| \leq 168\rEm$ holds simultaneously for every level $\ell = \ell_{\min}, \ell_{\min}+1,\dots,\log_3{n}$ is at least $1-\bigO\left(\dfrac{\log_3 n}{\mb^{1/3}}\right)$.
\end{lemma}

\begin{proof}
At the initial level $\ell_{\min}$, since $\rEm \geq 2\mb$, it follows that the size of the possibly defective set is at most
\begin{align*}
    |\mathcal{PD}^{(\ell_{\min})}| = \displaystyle\binom{3^{\ell_{\min}}}{3} \leq 3^{3\ell_{\min}} = 3^{3\lceil\log_3\mb^{1/3}\rceil} \leq 3^{3(\log_3{\mb^{1/3}}+1)} = 27\mb < 168\rEm
\end{align*}

We now prove the claim by induction. In particular, conditioned on $|\mathcal{PD}^{\ell}| \leq 168\rEm$, among the PD triples at the \( \ell \)-th level, at most \( \rEm \) are defective triples, which generate at most \( 84\rEm \) children at the next level.  
Denote $\mu = \EE_{\ell}\!\left(\displaystyle\sum_{u,v,w} \rE_{uvw} \right)$ and $\sigma^2 = \Var_{\ell} \!\left(\displaystyle\sum_{u,v,w} \rE_{uvw} \right)$ (Note that the sum is over all non-defective triples in $\mathcal{PD}^{(\ell)}$). By Lemma~\ref{lem:expectation} and Lemma~\ref{lem:variance}, we have
\begin{align*}
    \mu \leq \dfrac{\rEm}{2}, \quad \sigma^2 \leq \bigO\left( \frac{\rEm^2}{\mb^{1/3}} \right)
\end{align*}
Applying Chebyshev's inequality \eqref{eq:chebyshev}, we obtain
\begin{align*}
    \PP\left( \sum_{u,v,w} \rE_{uvw} > \rEm \right) \leq \PP\left( \left|\sum_{u,v,w} \rE_{uvw} - \mu\right| > \dfrac{\rEm}{2} \right) \leq \dfrac{\bigO\left( {\rEm^2}/{\mb^{1/3}} \right)}{\rEm^2/4} = \bigO\left(\dfrac{1}{\mb^{1/3}}\right)
\end{align*}

Therefore, with probability at least \( 1 - \bigO \left( \dfrac{1}{\mb^{1/3}}\right) \), at most \( \rEm \) non-defective triples are incorrectly retained as PD. These contribute at most another \( 84\rEm \) children at the next level, leading to a total of at most \( 84\rEm + 84\rEm = 168\rEm \) PD triples. As a result, we can obtain
\begin{align*}
    \PP\left(|\mathcal{PD}^{\ell+1}| \leq 168\rEm \middle| |\mathcal{PD}^{\ell}| \leq 168\rEm\right) \geq 1 - \bigO \left( \dfrac{1}{\mb^{1/3}}\right)
\end{align*}

Repeating the above argument inductively over the levels, the probability that the bound first fails at any given level is at most $\bigO(\mb^{-1/3})$.
Taking a union bound over the at most \( \log_3 n \) levels, we obtain
\begin{align*}
    \PP\left(
        \max_{\ell_{\min}\leq\ell\leq\log_3 n}
        |\mathcal{PD}^{(\ell)}|
        \leq
        168\rEm
    \right)
    \geq
    1-\bigO\left(
        \frac{\log_3 n}{\mb^{1/3}}
    \right),
\end{align*}
which completes the proof.
\end{proof}

We now prove the main theorem.
\begin{proof}[Proof of Theorem~\ref{thm:main}] The stated claims are inferred as follows.
\begin{itemize}
	\item { \textbf{Number of tests:} At each intermediate level from $\left\lceil \log_3 \overline{m}^{1/3}\right\rceil$ to $\log_3 n-1$, we use $C_1C_2\overline{m}$ tests. At the final level, we use $C_1C_2C'\overline{m}\log n$ tests. Thus, the total number of tests is
\begin{align*}
    C_1C_2\overline{m}
\left(
\log_3 n-\left\lceil \log_3 \overline{m}^{1/3}\right\rceil
\right)
+
C_1C_2C'\overline{m}\log n = \bigO(\overline{m}\log n).
\end{align*}

	}
	\item{ \textbf{Decoding time:} 
		From Eq.~\eqref{eq:definition-typical-set} used in Lemma~\ref{lem:typical}, there are at most \( 12\mb \) defective triples at each level for $\theta>2/3$ and there are at most $2\mb\log^2{\mb}$ defective triples at each level for $\theta \leq 2/3$. 
        
		At each level \( \ell \), for each possible defective triple, we conduct at most \( C_2 \mb^{2/3} \) outcome checks. This gives a total of \( \bigO \big( \mb^{2/3} |\mathcal{PD}^{(\ell)}| \big) \) outcome checks at level \( \ell \). Therefore, the total number of outcome checks from levels \( \ceil{\log_3{\mb^{1/3}}} \) to \( \log_3 n - 1 \) is 
		\[
		\bigO \left( \mb^{2/3} \sum_{\ell = \ceil{\log_3{\mb^{1/3}}}}^{\log_3 n - 1} |\mathcal{PD}^{(\ell)}| \right).
		\]
		At the final level \( \ell  = \log_3 n \), we conduct \( \bigO \big( \mb^{2/3} \log n \cdot |\mathcal{PD}^{(\ell)}| \big) \) outcome checks.
		
		For the case $\theta > \dfrac{2}{3}$, as shown in Lemma \ref{lem:bounds_size_PD}, with probability at least \( 1 - \bigO \left( \dfrac{\log_3{n}}{\mb^{1/3}}\right) \), we have \( |\mathcal{PD}^{(\ell)}| \leq 168\rEm = 2016\mb\) for all \( \ell \in \{\ceil{\log_3 \mb^{1/3}}, \dots, \log_3 n\} \). Therefore, the total number of outcome checks in the decoding procedure is at most \( \bigO(\mb^{5/3}\log n) \), with probability at least \( 1 - \bigO \left( \dfrac{\log_3{n}}{\mb^{1/3}}\right) \). 

        For the case $\theta \leq \dfrac{2}{3}$, replacing \( |\mathcal{PD}^{(\ell)}| = \bigO(\mb) \) by \( |\mathcal{PD}^{(\ell)}| = \bigO(\mb\log^2\mb) \) and proceeding exactly as above, we conclude that the total number of outcome checks in the decoding procedure is \( \bigO(\mb^{5/3}\log^2\mb\log n) \) with probability at least \( 1 - \bigO\!\left(\dfrac{\log_3 n}{\mb^{1/3}}\right) \).
	}
	\item{ \textbf{Error probability:}
    By Lemma \ref{lem:false_negative}, every true hyperedge is preserved throughout the decoding procedure, and hence $E\subseteq \widehat{E}$ deterministically. Therefore, it remains only to show that, with probability (1-o(1)), every non-defective triple is eliminated at the final level. By the law of total probability, we have 
    \begin{equation}
\begin{aligned}
\PP(\widehat{E}=E)
&\ge
\PP(H\in\mathcal{T}(\epsilon_n))
\PP(\widehat{E}=E\mid H\in\mathcal{T}(\epsilon_n)) \\
&\ge
\PP(H\in\mathcal{T}(\epsilon_n))
\PP\!\left(
|\mathcal{PD}^{(\log_3 n)}|
\le 168\rEm
\,\middle|\,
H\in\mathcal{T}(\epsilon_n)
\right) \\
&\qquad\cdot
\PP\!\left(
\widehat{E}=E
\,\middle|\,
|\mathcal{PD}^{(\log_3 n)}|
\le 168\rEm,\,
H\in\mathcal{T}(\epsilon_n)
\right).
\end{aligned}
\end{equation}
    From the previous analysis in Lemma \ref{lem:bounds_size_PD}, we have established that
    \begin{align*}
        \PP\left(\displaystyle\max_{\ell_{\min}\leq \ell \leq \log_3{n}}|\mathcal{PD}^{(\ell)}| \leq 168\rEm \middle| H \in \mathcal{T}(\epsilon_n) \right) \geq 1 - \bigO \left( \dfrac{\log_3{n}}{\mb^{1/3}}\right)
        \label{eq:error}
    \end{align*}
    
    At the final level, we conduct $C'\log n$ independent rounds, each consisting of $C_1 C_2 \mb$ tests. We analyze the error probability under the high probability condition that the size of the  $\mathcal{PD}^{(\log_3 n)}$ is at most $\bigO(\mb \log^2 \mb )$.
		
		For any fixed non-defective triple $(x,y,z)$ and a given sequence of $C_1 C_2 \mb$ tests, the probability that the triple is not  identified is at most $\exp(-(C_1 - 149))$, following a similar argument as in Lemma~\ref{lem:expectation}. Since we do $C'\log n$ repetitions, the probability that $(x,y,z)$ remains undetected is
        \begin{align*}
            \PP((x,y,z) \text{ is unidentified at the final level}) \leq \exp(-C'\log{n}(C_1-149)) = \dfrac{1}{n^{C'(C_1-149)}}
        \end{align*}
        Thus, for any constant $c > 0$, there exists a  choice of $C'$ such that the probability of any specific non-defective triple is not identified is at most $\bigO(n^{-c})$. Applying a union bound over the $|\mathcal{PD}^{(\log_3 n)}|$ non-defective triples at the final level, we then obtain
        \begin{align}
            \PP(\widehat{E} = E \mid |\mathcal{PD}^{\log_3{n}}| \leq 168\rEm, H \in \mathcal{T}(\epsilon_n)) \geq 1 - \bigO\left(\dfrac{\rEm}{n^{c}}\right)
        \end{align}
        Note that $\rEm = 12\mb$ (if $\dfrac{2}{3} < \theta < 1$) or $\rEm = 2\mb\log^2{\mb}$ (if $0 <\theta \leq \dfrac{2}{3}$) and $\mb = \Theta(n^{3\theta})$. In either cases, one can see that $\rEm \ll n^4$ as $n \to \infty$. As we have stated earlier in Lemma \ref{lem:typical}, $\PP(H \in \mathcal{T}(\epsilon_n)) = 1 - o(1)$ as $n \to \infty$. Substituting into \eqref{eq:error} yields
        \begin{align*}
            \PP(\widehat{E}=E) \geq 1 - o(1) - \bigO \left( \dfrac{\log_3{n}}{\mb^{1/3}}\right) - \bigO\left(\dfrac{\rEm}{n^{c}}\right) = 1 - o(1), \quad \text{for any } c > 4
        \end{align*}
	}
\end{itemize}

\end{proof}

\section{Discussion}
\subsection{Obstacles to General $k$-Uniform Extensions}
Although the hierarchical splitting framework extends naturally to the \(k\)-uniform setting, the resulting decoding procedure incurs a severe combinatorial blow-up. In this section, we briefly analyze this naive extension and show that the number of child candidate tuples grows exponentially in \(k\), leading to at least exponential dependence of the running time on \(k\). This explains why our main results focus on the 3-uniform case.
\[
\ell=\left\lceil \log_k \mb^{1/k}\right\rceil,\,
\left\lceil \log_k \mb^{1/k}\right\rceil+1,\dots,\log_k n.
\]
At each level $\ell$, the vertices are divided into $k^\ell$ disjoint blocks $\left\{\Hc^{(\ell)}_1,\dots,\Hc^{(\ell)}_{k^\ell}\right\},$ where every block contains exactly $n/k^\ell$ vertices. Equivalently, these blocks form the nodes of a $k$-ary hierarchical refinement tree, with each block at level $\ell$ being subdivided into $k$ child blocks at level $\ell+1$.

The test design is randomized. In each repetition, every level-$\ell$ block is assigned independently and uniformly at random to one of $C_1\mb^{1/k}$ tests. This random assignment procedure is carried out for $C_2\mb^{1-1/k}$ independent repetitions, where $C_1,C_2>0$ are constants specified later. Hence, the total number of repetitions scales on the order of $\mb^{1-1/k}$.

\begin{algorithm}[H]
	\caption{Testing Procedure}
	\label{alg:testing_general}
	\begin{algorithmic}[1]
		\REQUIRE Number of nodes \( n \), expected number of hyperedges \( \mb \), constants $C_1,C_2,C'$.
		\STATE Set \( \ell_{\min} \gets \ceil{\log_k \mb^{1/k}} \)
		\FOR{each level \( \ell = \ell_{\min}, \dotsc, \log_k n - 1 \)}
		\FOR{each iteration in \( \{1, \dotsc, C_2 \mb^{1-1/k}\} \)} 
        \STATE Initialize a sequence of $C_1 \overline{m}^{1/k}$ tests
		\FOR{each block \( j = 1, \dotsc, k^{\ell} \)}
		\STATE Assign block \( \Hc^{(\ell)}_j \) to a randomly chosen test among the $C_1 \mb^{1/k}$ tests
		\ENDFOR
		\ENDFOR
		\ENDFOR
		\STATE At the final level $\ell = \log_k n$, each block reduces to a singleton, i.e., $\Hc_{j}^{(\ell)} = \{j\}$. Repeat steps 3--6 for $C' \log n$ rounds, with each round consisting of $C_1 C_2 \mb$ tests.
	\end{algorithmic}
\end{algorithm}

As for the decoding procedure, the process is the same as well. At each level $\ell$ we keep a set of \emph{possible defective} blocks $\mathcal{PD}^{(\ell)}$, whose elements are block tuples $(\Hc_{i_1},\dots,\Hc_{i_k})$ that may still contain a true hyperedge at the current resolution. If a tuple \((\Hc_{i_1},\dots,\Hc_{i_k})\in\mathcal{PD}^{(\ell)}\) is not eliminated, we refine it as follows. Each block is partitioned into \(k\) children:
\[
\Hc_{i_r}\to \{\Hc_{i_r}^{(1)},\dots,\Hc_{i_r}^{(k)}\},
\qquad r=1,\dots,k.
\]
Let $\mathcal{C}_{i_1,\dots,i_k}
:=
\bigcup_{r=1}^{k}
\{\Hc_{i_r}^{(1)},\dots,\Hc_{i_r}^{(k)}\}.$ We then add to $\mathcal{PD}^{(\ell+1)}$ all tuples consisting of $k$ distinct
elements chosen from $\mathcal{C}_{i_1,\dots,i_j,\dots,i_k}$, namely
\[
\binom{\mathcal{C}_{i_1,\dots,i_j,\dots,i_k}}{k} = \binom{k^2}{k}
\]

Thus, the procedure operates solely by either discarding candidate tuples or refining them into smaller subtuples. At the final level, where each block has become a singleton vertex, every surviving vertex tuple that is not excluded by any negative test is declared to be a hyperedge.

\begin{algorithm}[H]
	\caption{Decoding Procedure}
	\label{alg:decoding_general}
	\begin{algorithmic}[1]
		\REQUIRE Test outcomes \( \{Y^{(i)}\}_{i=1}^t \) from Algorithm~\ref{alg:testing_general}, number of nodes \( n \), expected number of hyperedges \( \mb \)

		\STATE Initialize candidate set
		\[
		\mathcal{PD}^{(\ell_{\min})}
		=
		\left\{
		(\Hc^{(\ell_{\min})}_{i_1},\Hc^{(\ell_{\min})}_{i_2},\dots,\Hc^{(\ell_{\min})}_{i_k})
		:\;
		1 \le i_1 < i_2 < \cdots < i_k \le k^{\ell_{\min}}
		\right\}
		\]

		\FOR{each level \( \ell = \ell_{\min}, \dotsc, \log_k n - 1 \)}
			\FOR{each tuple \( (\Hc_{i_1},\Hc_{i_2},\dots,\Hc_{i_k}) \in \mathcal{PD}^{(\ell)} \)}
				\IF{no negative test contains \( \Hc_{i_1}, \Hc_{i_2}, \dots, \Hc_{i_k} \)}
					\STATE Add all child tuples generated from \( (\Hc_{i_1},\Hc_{i_2},\dots,\Hc_{i_k}) \) to \( \mathcal{PD}^{(\ell+1)} \)
				\ENDIF
			\ENDFOR
		\ENDFOR

		\STATE Let \( \widehat{E} \) be the set of tuples in \( \mathcal{PD}^{(\log_k n)} \) that are not ruled out by any negative test at the final level

		\STATE Return hypergraph estimate \( \widehat{H} = (V, \widehat{E}) \)
	\end{algorithmic}
\end{algorithm}

An analogous argument similar to the proof of Theorem \ref{thm:main} shows that, for any fixed value of $k$, the testing and decoding procedures described in Algorithm~\ref{alg:testing_general} and Algorithm~\ref{alg:decoding_general} achieve a decoding time of $\mathcal{O}(m^{2-1/k}\log n)$ while using $\mathcal{O}(m\log n)$ tests. However, when $k$ is treated as a variable parameter, the complexity incurs an additional exponential dependence on $k$. 

In particular, for each level $\ell \in \{\lceil \log_k \mb^{1/k}\rceil, \lceil \log_k \mb^{1/k}\rceil+1, \dots, \log_k n\}$, we define the following parameters:
\begin{itemize}
    \item $E_g:$ The set of defective hyperedges in the level block hypergraph.
    \item $v_g^{(i)}$: the number of $i$-tuples of blocks that collectively contain the vertices of at least one hyperedge.
    
    \item \(d_g^{(x,y)}(\mathcal H_{i_1},\dots,\mathcal H_{i_x})\): for a fixed non-defective \(x\)-tuple of blocks \((\mathcal H_{i_1},\dots,\mathcal H_{i_x})\), this denotes the number of \(y\)-tuples of distinct blocks \((\mathcal H_{j_1},\dots,\mathcal H_{j_y})\) such that: (i) the union of these \(x+y\) blocks can contain a hyperedge; and (ii) no proper subcollection of fewer than \(x+y\) blocks containing \((\mathcal H_{i_1},\dots,\mathcal H_{i_x})\) already forms a defective configuration counted in a lower-order case. We further define
    \[
    d_g^{(x,y)} \coloneqq \max_{\substack{(\mathcal{H}_{i_1}, \dots, \mathcal{H}_{i_x}) \text{non-defective}}}
    d_g^{(x,y)}(\mathcal{H}_{i_1}, \dots, \mathcal{H}_{i_x}).
    \]
\end{itemize}

By extending the counting and concentration arguments in Lemma \ref{lem:typical}, with probability $1-o(1)$, then: 
\begin{align*}
    \begin{cases}
        v^{(i)}_g \leq 3\mb^{i/k} \cdot \exp{(2k)} \cdot i^{k-i} \\[1ex]
        d_g^{(x,y)} \leq 3\mb^{y/k} \cdot \exp{(k+y)} \cdot k^k \\[1ex]
        |E_g| \leq E_{max} \coloneqq 2\mb \cdot \exp{(2k)} \cdot \displaystyle\max_{1 \leq i \leq k-1}{i^{k-i}} 
    \end{cases}
\end{align*}
Applying the same arguments as in Lemma~\ref{lem:expectation}, Lemma~\ref{lem:variance}, and Lemma~\ref{lem:bounds_size_PD}, we conclude that there exist positive constants \(C_1,C_2,C_k\), depending only on \(k\), such that with probability \(1-o(1)\),
\[
|\mathcal{PD}^{(\ell)}| \le C_k \binom{k^2}{k} E_{\max}.
\]
From the binomial lower bound,
\[
\binom{k^2}{k} \ge k^k = \exp(k\ln k),
\]
it follows that the number of surviving candidate tuples grows at least exponentially in \(k\). Therefore, the direct extension incurs exponential dependence on \(k\), revealing a fundamental limitation of the approach.

\subsection{Possible Decoding Acceleration via Permutation Techniques}

A natural question is whether the decoding time can be further improved by incorporating a permutation-based amplification step, analogous to techniques used in related graph-learning frameworks.

In particular, suppose that for each induced subhypergraph there exists a suitable collection of random permutations such that at least one permutation yields a favorable block configuration for the hierarchical splitting decoder (that is, there exists no defective block and no defective pair of blocks in the subhypergraphs). Then, one would obtain a non-adaptive recovery algorithm using $O\!\left(m^{1+\frac{9\delta}{26}}\log n\right)$ tests and achieving decoding time $O\!\left(m^{1+\delta}\log n\right),$ for any fixed constant $\delta>0$. The mildly superlinear increase in the number of tests arises from the need to perform a sufficiently large collection of tests across multiple candidate permutations, in order to identify one that eliminates all defective pairs of blocks. 

Nevertheless, such an improvement in decoding time can only be obtained under the assumption that there exists a family of \(\varepsilon\)-almost \(3\)-wise independent permutations with error parameter \(\varepsilon=\bigO(N^{-3})\), while still supporting forward and inverse evaluation in near-constant time. Achieving these properties simultaneously appears highly nontrivial. This also illustrates the difficulty of extending the permutation technique beyond the graph setting to \(3\)-uniform or higher-uniform hypergraphs, as it heavily relies on a yet-to-be-constructed permutation family that simultaneously provides strong pseudorandomness and efficient evaluation.
\section{Conclusion}
We extend the fast binary splitting framework for graph learning to the non-adaptive setting of Erd\H{o}s--R\'enyi 3-uniform hypergraphs. Our approach achieves a test complexity of $\mathcal{O}(\mb \log n)$ and a decoding time of $\mathcal{O}(\mb^{5/3}\log^2{\mb}\log n)$.

Several directions for future work remain open. While we provide a partial extension to $k$-uniform hypergraphs, the resulting complexity exhibits exponential dependence on $k$. Designing algorithms that achieve efficient scaling for growing $k$ while retaining the optimal order in $n$ and $\mb$  remains an open problem. Another promising direction is to extend the framework to noisy settings, where queries may return erroneous responses.

\subsection*{Acknowledgement}
We would like to thank Jonathan Scarlett for useful discussions about hypergraph learning, as well as the binary splitting method.

\bibliographystyle{alpha}
\bibliography{bibliofile}

\newpage
\appendix

\section{Some probabilistic inequalities}

In this section, we present several probabilistic inequalities that will be used throughout the paper.

\paragraph{Markov's inequality.}
Let $X$ be a non-negative random variable and let $t>0$. Then
\begin{align}
    \PP(X \ge t) \le \frac{\EE[X]}{t}.
    \label{eq:markov}
\end{align}

\paragraph{Chebyshev's inequality.}
Let $X$ be a random variable with finite mean $\mu = \EE[X]$ and variance $\sigma^2 = \Var(X)$. 
For any $t>0$, we have
\begin{align}
    \PP(|X - \mu| \ge t) \le \frac{\Var(X)}{t^2}.
    \label{eq:chebyshev}
\end{align}
\subsection{Concentration inequalities}
\label{appendix:concentration}

Let $X_1, X_2, \dots, X_n$ be a sequence of independent $\mathrm{Bernoulli}(q)$ random variables. 
Denote $X = \sum_{i=1}^{n} X_i$ and $\mu \coloneqq qn$. 
We then have the following:
\begin{enumerate}
    \item For $\delta > 0$, 
    \begin{equation}
      \PP \!\left(X \ge (1+\delta)\mu\right) 
      \leq \exp \!\left(-\frac{\delta^2 \mu}{2+\delta} \right).
      \label{eq:chernoff1}
    \end{equation}
    \item For any $t \geq \mu$, 
    \begin{equation}
      \PP(X \ge t) 
      \leq \left( \frac{e\mu}{t} \right)^{t}.
      \label{eq:chernoff2}
    \end{equation}
\end{enumerate}

\section{Proof of Lemma~\ref{lem:typical} (Typical Set Probability)}
\label{sec:appendix_lemm_typical}
 
As for condition (i), since $M$ follows a binomial distribution with mean $\mb = q\cdot\binom{n}{3}$; thus, it holds by a standard concentration inequalities argument.

We now turn to condition (ii). For each level $\ell$, let $g = 3^\ell$ and consider the quantities $\nu^1_g, \nu^2_g, d^{(1,1)}_g, d^{(1,2)}_g, d^{(2,1)}_g$. These are random variables determined by $G$, and we prove that they satisfy the bounds in \ref{lem:typical} with probability at least $1 - o(1)$ as $n \to +\infty$. Let $s = n/g$ denote the block size. Before dividing into the proof, we first note that 
\begin{align*}
    \mb^{1/3} \leq g \leq n
\end{align*}

\subsection*{Bounding $\nu_g^1$}

For each block $\Hc_u$, in which $u \in [g]$, define $Z_u$ to be the Bernoulli random variable:
\begin{align*}
    Z_u = \begin{cases}
        1 & \exists\, e \in E(H): e \subseteq \Hc_u \\
        0 & \text{otherwise}
    \end{cases}
\end{align*}
Since we have defined $\nu^1_g$ as the number of defective blocks, it follows that
\begin{align*}
    \nu^1_g = \displaystyle\sum_{u \in [g]}{Z_u}
\end{align*}

Each block $\Hc_u$ contains $|\Hc_u| = \dfrac{n}{g}$ vertices; hence, the number of possible hyperedges inside $\Hc_u$ is $\displaystyle\binom{n/g}{3}$. Therefore, we have:
\begin{align*}
    \mathbb{P}(Z_{u}=1) = 1 - (1-q)^{\binom{n/g}{3}}
\end{align*}
Using the identity $1 - (1-q)^T \leq qT$, we can obtain
\begin{align*}
\mathbb{P}(Z_u=1) = 1 - (1-q)^{\binom{n/g}{3}} \leq q\binom{n/g}{3}
\end{align*}
Combining with the fact that $\mb = q \cdot \binom{n}{3}$, we can obtain:
\begin{align*}
    \mathbb{P}(Z_u = 1) \leq \overline{m} \cdot \dfrac{\binom{n/g}{3}}{\binom{n}{3}} = \dfrac{\overline{m}}{g^3}\cdot\dfrac{n-g}{n-1}\cdot\dfrac{n-2g}{n-2}
\end{align*}
Since $g \geq 1$, we have $\dfrac{n-g}{n-1} \leq 1$ and $\dfrac{n-2g}{n-2} \leq 1$, which then implies
\begin{align*}
    \mathbb{P}(Z_u = 1) \leq \dfrac{\overline{m}}{g^3}\cdot\dfrac{n-g}{n-1}\cdot\dfrac{n-2g}{n-2} \leq \dfrac{\overline{m}}{g^3}
\end{align*}
By linearity of expectation and the fact that $\overline{m}^{\textstyle\frac{1}{3}}\leq g \leq n$, one obtains:
\begin{align*}
     \mu_{\nu^1_g} = \mathbb{E}(\nu^1_g) = \mathbb{E}\left(\displaystyle\sum_{u \in [g]}{Z_u}\right) \leq \dfrac{\overline{m}}{g^2} \leq \overline{m}^{\textstyle\frac{1}{3}}
\end{align*}

Since the blocks are pairwise disjoint and hyperedges are generated independently, the random variables $\{Z_u\}_{u \in [g]}$ are independent. We now consider two cases: 
\begin{itemize}
    \item \textbf{Case 1:} $\theta > \dfrac{2}{3}$. Since $\overline{m} = \Theta(n^{3\theta})$ and $\overline{m}^{1/3} \leq g \leq n$, there exists a positive constant $c_1>0$ such that for all sufficiently large $n$, we have
    \begin{align*}
        \dfrac{\mb}{g^2} \in \left[c_1n^{3\theta-2},\mb^{1/3}\right]
    \end{align*}
    Applying the Chernoff inequality \eqref{eq:chernoff2} yields
    \begin{align}
\mathbb{P}\left(\nu^1_g > \dfrac{3\mb}{g^2}\right)  \leq \left(\dfrac{e\mu_{\nu^1_g}}{3\mb/g^2}\right)^{3\mb/g^2} \leq \left(\dfrac{e}{3}\right)^{3\mb/g^2} = \exp{\left(-\Omega(n^{3\theta-2})\right)}
\label{bound-v_1-1}
\end{align}
    \item \textbf{Case 2:} $\theta \leq \dfrac{2}{3}$. From the Chernoff inequality \eqref{eq:chernoff2}, it is immediate that:
\begin{align}
\mathbb{P}\left(\nu^1_g > 3\overline{m}^{1/3}\right)  \leq \left(\dfrac{e\mu_{\nu^1_g}}{3\overline{m}^{1/3}}\right)^{3\overline{m}^{1/3}} \leq \left(\dfrac{e}{3}\right)^{3\overline{m}^{1/3}} =  \exp{\left(-\Theta(\mb^{1/3})\right)} = \exp{(-\Theta(n^{\theta}))}
\label{bound-v_1-2}
\end{align}
\end{itemize}

Since both bounds in \eqref{bound-v_1-1} and \eqref{bound-v_1-2} decay to zero (significantly) faster than $\dfrac{1}{\log_3{n}}$, meaning that they are still $o(1)$ after a union bound over at most $\log_3{n}$ values of $\ell$. Thus, with probability $1-o(1)$, \eqref{eq:v1} holds for all $\ell$.

\subsection*{Bounding $\nu^2_g$}
Let $T$ denote the set of hyperedges $e = \{a,b,c\}$ for which there exist distinct indices
$x,y \in [g]$, with $x < y$, such that
\[
\{a,b,c\} \subseteq \mathcal{H}_x \cup \mathcal{H}_y
\quad \text{and} \quad
\{a,b,c\} \nsubseteq \mathcal{H}_x,\;
\{a,b,c\} \nsubseteq \mathcal{H}_y .
\]
For each hyperedge $e \in T$, define the mapping
\[
\phi(e) := \{x,y\},
\]

where $\{x,y\}$ is the (unordered) pair of block indices such that $e$ intersects both
$\mathcal{H}_x$ and $\mathcal{H}_y$. This mapping is well defined, since the blocks form a partition of $V$, every vertex belongs to a unique block. Hence the unordered set of block indices intersected by a hyperedge is uniquely determined. Recall that $\nu_g^{(2)}$ is the number of defective block pairs whose individual
blocks are both non-defective. For any defective block pair $(\mathcal{H}_x,\mathcal{H}_y)$ with both blocks
non-defective, there exists at least one hyperedge $e \in T$ such that $\phi(e)=\{x,y\}$. Hence, every such defective block pair lies in the image of $\phi$, and we have
\[
\nu_g^{2} \le |\phi(T)| \leq |T|
\]

We now bound $|T|$. For a fixed block pair $\{x,y\}$ with $x \neq y$, let $s = n/g$ denote the block size.
The number of vertex triples contained in $\mathcal{H}_x \cup \mathcal{H}_y$ but not entirely within a
single block is
\[
\binom{2s}{3} - 2\binom{s}{3} = \dfrac{2s(2s-1)(2s-2)}{6} - \dfrac{2s(s-1)(s-2)}{6} = s^2(s-1)
\]
Let $S$ denote the set of all vertex triples $\{a,b,c\}$ such that
\[
\{a,b,c\} \subseteq \mathcal{H}_x \cup \mathcal{H}_y
\quad \text{and} \quad
\{a,b,c\} \nsubseteq \mathcal{H}_x,\;
\{a,b,c\} \nsubseteq \mathcal{H}_y
\]
for some distinct indices $x,y \in [g]$. Since there are $\binom{g}{2}$ choices of block pairs, we have
\[
|S| = \binom{g}{2}s^2(s-1).
\]
For each triple $\{a,b,c\} \in S$, define the Bernoulli random variable
\[
X_{a,b,c} =
\begin{cases}
1, & \text{if } \{a,b,c\} \in E(H), \\
0, & \text{otherwise}.
\end{cases}
\]
Under the Erd\H{o}s--R\'enyi hypergraph model, the random variables
$\{X_{a,b,c}\}_{\{a,b,c\}\in S}$ are independent and satisfy
\[
\mathbb{P}(X_{a,b,c} = 1) = q.
\]
Recalling that $T$ is precisely the set of hyperedges crossing exactly two blocks, we have
\[
|T| = \sum_{\{a,b,c\}\in S} X_{a,b,c}
\sim \mathrm{Bin}(|S|, q).
\]
By linearity of expectation,
\begin{align*}
\mu_T & = \mathbb{E}(|T|)
= q \binom{g}{2}s^2(s-1) 
\le \frac{g^2}{2} \cdot q \cdot s^2(s-1)
= \frac{g^2}{2} \cdot \frac{\overline{m}}{\binom{n}{3}} \cdot s^2(s-1) \\
& = \dfrac{g^2}{2} \dfrac{6\mb}{n(n-1)(n-2)}\dfrac{n^2(n-g)}{g^3}
= \dfrac{3\mb}{g} \dfrac{n(n-g)}{(n-1)(n-2)}
\end{align*}

For sufficiently large $n$, we have $g \ge \mb^{1/3} \geq 3$. As a result, $(n-1)(n-2) = n^2 - 3n + 2 \ge n^2 - gn = n(n-g)$ which then implies that
\begin{align*}
\mu_T \leq \dfrac{3\mb}{g} \leq 3\mb^{2/3}
\end{align*}
We now consider two cases:
\begin{itemize}
    \item \textbf{Case 1:} $\theta > \dfrac{2}{3}$. Since $\overline{m} = \Theta(n^{3\theta})$ and $\overline{m}^{1/3} \leq g \leq n$, there exists a constant $c_1>0$ such that for all sufficiently large $n$, we have
    \begin{align*}
        \dfrac{\mb}{g} \in \left[c_1n^{3\theta-1},\mb^{2/3}\right]
    \end{align*}
    Applying a Chernoff bound in \eqref{eq:chernoff2}, one can see that
\begin{align}
    \mathbb{P}\!\left(\nu_g^{2} > \frac{9\overline{m}}{g}\right)
\le
\mathbb{P}\!\left(|T| > \frac{9\overline{m}}{g}\right)
\le \left(\dfrac{e\mu_T}{9\mb/g}\right)^{9\mb/g}
\le
\left(\frac{e}{3}\right)^{9\mb/g} = \exp{\left(-\Omega(n^{3\theta-1})\right)}
\label{eq:bound-v-2-1}
\end{align}
    \item \textbf{Case 2:} $\theta \leq \dfrac{2}{3}$. Applying a Chernoff bound in \eqref{eq:chernoff2}, we obtain
\begin{align}
    \mathbb{P}\!\left(\nu_g^{2} > 9\overline{m}^{2/3}\right)
\le
\mathbb{P}\!\left(|T| > 9\overline{m}^{2/3}\right)
\le
\left(\dfrac{e\mu_T}{9\mb^{2/3}}\right)^{9\overline{m}^{2/3}}
\le
\left(\frac{e}{3}\right)^{9\overline{m}^{2/3}} = \exp{\left(-\Theta(n^{2\theta})\right)}
\label{eq:bound-v-2-2}
\end{align}
\end{itemize}

Since the bounds in \eqref{eq:bound-v-2-1} and \eqref{eq:bound-v-2-2} approach 0 strictly faster than $\dfrac{1}{\log_3{n}}$ when $n \to +\infty$, we may take a further union bound over all levels $\ell$ (with at most $\log_3{n}$ levels). It follows that, with probability $1 - o(1)$, \eqref{eq:v2} holds for all $\ell$.

\subsection*{Bounding $|E_g|$}
From the definition of $\Hc_g$ and the fact that $\Hc$ contains at most $M$ edges, we have
\begin{align*}
    |E_g| \leq \binom{g-1}{2} \cdot \nu^1_g + (g-2) \cdot \nu^2_g + M \leq \dfrac{g^2}{2}\cdot \nu^1_g + g \cdot \nu^2_g + M
\end{align*}
We now consider two cases:
\begin{itemize}
    \item \textbf{Case 1:} $\theta > \dfrac{2}{3}$. 
    If $\nu^1_g \le \dfrac{3\mb}{g^2}$, $\nu^2_g \le \dfrac{9\mb}{g}$ and $M \le \dfrac{3\mb}{2}$, then we have $|E_g| \le 12\mb$, which implies that
    \begin{align*}
        \PP\left(\left(\nu^1_g \le \dfrac{3\mb}{g^2}\right) \cap \left(\nu^2_g \le \dfrac{9\mb}{g}\right) \cap \left(M \le \dfrac{3\mb}{2}\right)\right) \leq \PP(|E_g| \le 12\mb) \\ 
        1 - \PP(|E_g| \le 12\mb)  \leq 1 - \PP\left(\left(\nu^1_g \le \dfrac{3\mb}{g^2}\right) \cap \left(\nu^2_g \le \dfrac{9\mb}{g}\right) \cap \left(M \le \dfrac{3\mb}{2}\right)\right)
    \end{align*}
    Which is equivalent to 
    \begin{align*}
        \PP(|E_g| > 12\mb) \leq \PP\left(\left(\nu^1_g > \dfrac{3\mb}{g^2}\right) \cup \left(\nu^2_g > \dfrac{9\mb}{g}\right) \cup \left(M > \dfrac{3\mb}{2}\right)\right) \\ \leq \PP\left(\nu^1_g > \dfrac{3\mb}{g^2}\right) + \PP\left(\nu^2_g > \dfrac{9\mb}{g}\right) + \PP\left(M > \dfrac{3\mb}{2}\right)
    \end{align*}
    The last inequality comes from the union bound. From the Chernoff inequality \eqref{eq:chernoff1}, this gives
    \begin{align}
        \PP\left(M > \dfrac{3\mb}{2}\right) \leq \exp{\left(-\dfrac{\mb}{10}\right)} = \exp{\left(-\Omega(n^{3\theta})\right)}
        \label{eq-bound-m}
    \end{align}
    Combining \eqref{bound-v_1-1}, \eqref{eq:bound-v-2-1} and \eqref{eq-bound-m}, we have
    \begin{align}
        \PP\left(|E_g| > 12\mb\right) \leq \exp{\left(-\Omega(n^{3\theta-2})\right)} + \exp{\left(-\Omega(n^{3\theta-1})\right)} + \exp{\left(-\Omega(n^{3\theta})\right)}
        \label{eq:bound-e-g-1}
    \end{align}
    \item \textbf{Case 2:} $\theta \leq \dfrac{2}{3}$.
    Taking expectation, we obtain
\begin{align*}
    \mathbb{E}(|E_g|) \leq \dfrac{g^2}{2} \mathbb{E}(\nu^1_g) + g\mathbb{E}(\nu^2_g) + \mb \leq \dfrac{g^2}{2}\dfrac{\mb}{g^2} + g\dfrac{3\mb}{g} + \mb = \dfrac{9\mb}{2}
\end{align*}
By the Markov inequality \eqref{eq:markov}, we have
\begin{align}
    \mathbb{P}\left(|E_g| > 2\mb\log^2{\mb} \right) \leq \dfrac{9}{4\log^2{\mb}} = \bigO\left(\dfrac{1}{\log^2{n}}\right)
    \label{eq:bound-e-g-2}
\end{align}
\end{itemize}

We can see that all of the bounds in \eqref{eq:bound-e-g-1} and \eqref{eq:bound-e-g-2} decay to zero significantly faster than $\dfrac{1}{\log_3{n}}$; thereby implying that they are still $o(1)$ after taking a union bound over at most $\log_3{n}$ values of $\ell$. Hence, with probability $1 - o(1)$, \eqref{eq:eg} holds for all $\ell$.

\subsection*{Bounding $d^{(1,1)}_g$}
\label{subsection-bound-d11}
Fix a non-defective block $\mathcal{H}_i$, and let $s = n/g$ denote the block size.
For each $j \in [g]\setminus\{i\}$, consider the set of vertex triples that contain at
least one vertex from $\mathcal{H}_i$ and at least one vertex from $\mathcal{H}_j$,
but are not entirely contained in either block.
Let $S_i$ denote the union of all such triples over $j \neq i$, that is, the set of
triples $\{u,v,w\}$ such that either $\{u,v\} \subseteq \mathcal{H}_i$ and $w\in \mathcal{H}_j$, or $u\in \mathcal{H}_i$ and $\{v,w\} \subseteq \mathcal{H}_j$, for some $j\in[g]\setminus\{i\}$. Each triple in $S_i$ corresponds to a potential hyperedge that intersects
$\mathcal{H}_i$ and exactly one other block. For a fixed $j\neq i$, the number of such triples is
\[
\binom{2s}{3}-2\binom{s}{3}=s^2(s-1),
\]
and since there are $g-1$ choices of $j$, we obtain
\[
|S_i|=(g-1)\cdot s^2(s-1).
\]
For each $\{u,v,w\}\in S_i$, define the Bernoulli random variable
\[
X_{u,v,w}=
\begin{cases}
1, & \text{if } \{u,v,w\}\in E(H),\\
0, & \text{otherwise}.
\end{cases}
\]
Let
\[
X_i^{(1,1)} := \sum_{\{u,v,w\}\in S_i} X_{u,v,w},
\]

which counts the total number of hyperedges that intersect $\mathcal{H}_i$
and exactly one other block.
Under the Erd\H{o}s--R\'enyi hypergraph model, the random variables
$\{X_{u,v,w}\}_{\{u,v,w\}\in S_i}$ are independent and satisfy
$\mathbb{P}(X_{u,v,w}=1)=q$. By linearity of expectation,
\begin{align*}
\mu_{X_i^{(1,1)}} & = \mathbb{E}(X_i^{(1,1)})
= q\,|S_i|
= \dfrac{\mb}{\binom{n}{3}}\,(g-1)s^2(s-1) = \dfrac{6\mb}{n(n-1)(n-2)} \cdot (g-1)\dfrac{n^2}{g^2}\cdot\dfrac{n-g}{g} \\
& = \dfrac{6\mb}{g^2} \cdot \dfrac{g-1}{g} \cdot \dfrac{n}{n-1} \cdot \dfrac{n-g}{n-2} 
\end{align*}
For sufficiently large $n$, we have $2 \leq \mb^{1/3} \leq g \leq n$, which implies $\dfrac{g-1}{g} \leq \dfrac{n-1}{n}$. Consequently,
\begin{align*}
    \mu_{X_i^{(1,1)}} = \dfrac{6\mb}{g^2} \cdot \dfrac{g-1}{g} \cdot \dfrac{n}{n-1} \cdot \dfrac{n-g}{n-2}  \leq \dfrac{6\mb}{g^2} \leq 6\mb^{1/3}
\end{align*}
Let $T_i$ denote the (random) set of hyperedges realized among the triples in $S_i$,
that is,
\[
T_i := \bigl\{\, e=\{u,v,w\}\in E(H) : \{u,v,w\}\in S_i \,\bigr\}.
\]
For each hyperedge $e\in T_i$, define the mapping
\[
\phi_i(e) := j,
\]

where $j\neq i$ is the unique block index such that $e$ intersects $\mathcal{H}_j$.
This mapping is well defined, since $\{\mathcal{H}_1,\dots,\mathcal{H}_g\}$ is a
partition of $V$ into disjoint blocks, and each hyperedge $e\in T_i$ intersects
$\mathcal{H}_i$ and exactly one other block. Recall that $d_g^{(1,1)}(\mathcal{H}_i)$ denotes the number of indices $j \neq i$
such that the block pair $(\mathcal{H}_i,\mathcal{H}_j)$ is defective. It is clear that for every defective block pair $(\mathcal{H}_i,\mathcal{H}_j)$ with
$j \neq i$, there exists at least one hyperedge $e \in T_i$ such that $\phi_i(e)=j$.
Hence, every such index $j$ lies in the image of $\phi_i$, and we obtain
\[
d_g^{(1,1)}(\mathcal{H}_i) \le |\phi_i(T_i)|,
\]

where $\phi_i(T_i) \subseteq [g]\setminus\{i\}$ denotes the set of block indices
corresponding to defective block pairs involving $\mathcal{H}_i$. Since $\phi_i$ is a mapping from $T_i$ to $[g]\setminus\{i\}$, we always have
$|\phi_i(T_i)| \le |T_i| = X_i^{(1,1)}$. Therefore,
\[
d_g^{(1,1)}(\mathcal{H}_i) \le X_i^{(1,1)}.
\]
From the Chernoff bound \eqref{eq:chernoff2}, we can see that
\begin{align*}
    \mathbb{P}\left(d^{(1,1)}_g(\mathcal{H}_i) > 18\overline{m}^{1/3}\right) 
    \leq \mathbb{P}\left(X^{(1,1)}_i > 18\overline{m}^{1/3}\right) \leq \left(\dfrac{e\mu_{X^{(1,1)}_i}}{18\mb^{1/3}}\right)^{18\overline{m}^{1/3}} \leq \left(\dfrac{e}{3}\right)^{18\overline{m}^{1/3}}
\end{align*}
Finally, by a union bound over all $i \in [g]$, it follows that
\begin{align*}
    \mathbb{P}\left(d^{(1,1)}_g > 18\overline{m}^{1/3}\right) 
    = \mathbb{P}\left(\displaystyle\max_{i \in [g]}d^{(1,1)}_g(\mathcal{H}_i) > 18\overline{m}^{1/3}\right) 
    \leq \displaystyle\sum_{i \in [g]}{\mathbb{P}\left(d^{(1,1)}_g(\mathcal{H}_i) > 18\overline{m}^{1/3}\right)} \leq g\cdot\left(\dfrac{e}{3}\right)^{18\overline{m}^{1/3}}
\end{align*}

Because $g\cdot\left(\dfrac{e}{3}\right)^{18\overline{m}^{1/3}} \to 0$ significantly faster than $\dfrac{1}{\log_3{n}}$ when $n \to +\infty$, we may take a further union bound over all levels $\ell$ (with at most $\log_3{n}$ levels). It follows that with probability $1 - o(1)$, \eqref{eq:d11} holds for all values of $\ell$.

\subsection*{Bounding $d^{(1,2)}_g$}
\label{subsection:bound-d_12}
Fix a block $\mathcal{H}_i$ and let $s = n/g$ denote the block size.
Let $S_i^{(1,2)}$ denote the set of vertex triples $\{u,v,w\}$ such that: $u \in \mathcal{H}_i$, $v \in \mathcal{H}_{i_1}$ and $w \in \mathcal{H}_{i_2}$ for some $i_1 \neq i_2$ with $i \notin \{i_1,i_2\}$. Each such triple corresponds to a potential hyperedge that uses exactly one
vertex from $\mathcal{H}_i$ and one vertex from each of two other blocks.
We have:
\[
\left|S_i^{(1,2)}\right| = \binom{g-1}{2} \, s^3.
\]
For each triple $\{u,v,w\} \in S_i^{(1,2)}$, define the Bernoulli random variable
\[
X_{u,v,w} =
\begin{cases}
1, & \text{if } \{u,v,w\} \in E(H), \\
0, & \text{otherwise}.
\end{cases}
\]
At the same time, define
\[
X_i^{(1,2)} := \sum_{\{u,v,w\} \in S_i^{(1,2)}} X_{u,v,w},
\]

which counts the total number of hyperedges that use exactly one vertex from
$\mathcal{H}_i$ and one vertex from each of two distinct other blocks.
By the Erd\H{o}s--R\'enyi hypergraph model, the random variables $\{X_{u,v,w}\}$ are
independent and satisfy $\mathbb{P}(X_{u,v,w}=1)=q$. Taking expectation,
\begin{align}
\mu_{X_i^{(1,2)}} = \mathbb{E}\!\left(X_i^{(1,2)}\right)
= q \cdot \left|S_i^{(1,2)}\right|
= q \binom{g-1}{2} s^3 = \dfrac{3\mb}{g} \cdot \dfrac{g-1}{g} \cdot \dfrac{g-2}{g} \cdot \dfrac{n^2}{(n-1)(n-2)} 
\label{eq:bound-d12-1}
\end{align}
Since $g \leq n$, we have $\dfrac{g-1}{g} \leq \dfrac{n-1}{n}$ and $\dfrac{g-2}{g} \leq \dfrac{n-2}{n}$. Combining this with \eqref{eq:bound-d12-1} yields 
\begin{align*}
    \mu_{X_i^{(1,2)}} \leq  \dfrac{3\mb}{g} \cdot \dfrac{n-1}{n} \cdot \dfrac{n-2}{n} \cdot \dfrac{n^2}{(n-1)(n-2)}  = \dfrac{3\mb}{g} \leq 3\mb^{2/3}
\end{align*}
Recall that $d_g^{(1,2)}(\mathcal{H}_i)$ denotes the number of block pairs
$(\mathcal{H}_{i_1},\mathcal{H}_{i_2})$ such that the triple
$(\mathcal{H}_i,\mathcal{H}_{i_1},\mathcal{H}_{i_2})$ is defective, while
$\mathcal{H}_i$, $\mathcal{H}_{i_1}$, $\mathcal{H}_{i_2}$ and all corresponding
block pairs are non-defective. If a block pair $(\mathcal{H}_{i_1},\mathcal{H}_{i_2})$ contributes to
$d_g^{(1,2)}(\mathcal{H}_i)$, then by definition there exists at least one
hyperedge using exactly one vertex from each of the blocks
$\mathcal{H}_i$, $\mathcal{H}_{i_1}$, and $\mathcal{H}_{i_2}$.
Such a hyperedge is counted by $X_i^{(1,2)}$; thus, by the same argument used for
$d_g^{(1,1)}(\mathcal{H}_i)$, we conclude that
\[
d_g^{(1,2)}(\mathcal{H}_i) \le X_i^{(1,2)}.
\]
Applying a Chernoff bound \eqref{eq:chernoff2}, we obtain
\[
\mathbb{P}\!\left( d_g^{(1,2)}(\mathcal{H}_i) > 9\overline{m}^{2/3} \right)
\le
\mathbb{P}\!\left( X_i^{(1,2)} > 9\overline{m}^{2/3} \right)
\le 
\left(\dfrac{e\mu_{X_i^{(1,2)}}}{9\mb^{2/3}}\right)^{9\overline{m}^{2/3}}
\le
\left(\frac{e}{3}\right)^{9\overline{m}^{2/3}}.
\]
Finally, by a union bound over all $i \in [g]$,
\begin{align*}
\mathbb{P}\!\left(d_g^{(1,2)} > 9\overline{m}^{2/3}\right)
&=
\mathbb{P}\!\left(\max_{i\in[g]} d_g^{(1,2)}(\mathcal{H}_i) > 9\overline{m}^{2/3}\right) \le
\sum_{i\in[g]} \mathbb{P}\!\left( X_i^{(1,2)} > 9\overline{m}^{2/3} \right) \le
g \left(\frac{e}{3}\right)^{9\overline{m}^{2/3}}.
\end{align*}

As $g\left(\dfrac{e}{3}\right)^{9\overline{m}^{2/3}} \to 0$ much faster than $\dfrac{1}{\log_3{n}}$ when $n \to +\infty$, we may take a further union bound over all levels $\ell$ (with at most $\log_3{n}$ levels). It follows that with probability $1 - o(1)$, \eqref{eq:d12} holds for all $\ell$.

\subsection*{Bounding $d^{(2,1)}_g$}
\label{subsection:bound-d_21}
Fix a block pair $\{\mathcal{H}_{i_1},\mathcal{H}_{i_2}\}$ with
$1 \le i_1 < i_2 \le g$, and let $s = n/g$ denote the block size.
For each $j \in [g]\setminus\{i_1,i_2\}$, consider the set $S_{i_1,i_2}^{(2,1)}$ of
vertex triples $\{u,v,w\}$ such that: $u \in \mathcal{H}_{i_1}$, $v \in \mathcal{H}_{i_2}$, and $w \in \mathcal{H}_j$.

Each such triple corresponds to a potential hyperedge that uses exactly one
vertex from each of the blocks
$\mathcal{H}_{i_1}$, $\mathcal{H}_{i_2}$, and $\mathcal{H}_j$. We have:
\begin{align*}
    \left| S_{i_1,i_2}^{(2,1)} \right| = (g-2) \cdot s^3 
\end{align*}
For each $\{u,v,w\}\in S_{i_1,i_2}^{(2,1)}$, define the Bernoulli random variable
\[
X_{u,v,w} =
\begin{cases}
1, & \text{if } \{u,v,w\} \in E(H), \\
0, & \text{otherwise}.
\end{cases}
\]
At the same time, define
\[
X_{i_1,i_2}^{(2,1)} := \sum_{\{u,v,w\} \in S_{i_1,i_2}^{(2,1)}} X_{u,v,w},
\]

which counts the total number of hyperedges that use exactly one vertex from each of
the blocks $\mathcal{H}_{i_1}$, $\mathcal{H}_{i_2}$, and some third block
$\mathcal{H}_j$.
By the Erd\H{o}s--R\'enyi hypergraph model, these random variables are independent
and satisfy:
\begin{align*}
    \mathbb{P}(X_{u,v,w}=1)=q
\end{align*}
By linearity of expectations, we have
\begin{align}
\mu_{X_{i_1,i_2}^{(2,1)}} = \mathbb{E}\!\left( X_{i_1,i_2}^{(2,1)} \right)
& =
(g-2)\, q s^3 = (g-2) \cdot \frac{\overline{m}}{\binom{n}{3}} \cdot \frac{n^3}{g^3} = \dfrac{6\mb}{g^2} \cdot \dfrac{g-2}{g} \cdot \dfrac{n^2}{(n-1)(n-2)}
\label{eq:bound-d21-1}
\end{align}
Since $g\le n$, we have $\dfrac{g-2}{g} \leq \dfrac{n-2}{n}$, and for sufficiently large $n$, we have $\dfrac{n}{n-1} \leq \dfrac{13}{12}$. Combining these inequalities with \eqref{eq:bound-d21-1} yields
\begin{align*}
    \mu_{X_{i_1,i_2}^{(2,1)}} \leq \dfrac{6\mb}{g^2} \cdot \dfrac{13}{12} = \dfrac{13\mb}{2g^2} \leq \dfrac{13\mb^{1/3}}{2}
\end{align*}

Recall that $d_g^{(2,1)}(\mathcal{H}_{i_1},\mathcal{H}_{i_2})$ denotes the number of
indices $j$ such that the triple
$(\mathcal{H}_{i_1},\mathcal{H}_{i_2},\mathcal{H}_j)$ is defective, while
$\mathcal{H}_{i_1}$, $\mathcal{H}_{i_2}$, $\mathcal{H}_j$ and all corresponding
block pairs are non-defective. If a block $\mathcal{H}_j$ contributes to
$d_g^{(2,1)}(\mathcal{H}_{i_1},\mathcal{H}_{i_2})$, then by definition there exists
at least one hyperedge using exactly one vertex from each of the blocks
$\mathcal{H}_{i_1}$, $\mathcal{H}_{i_2}$, and $\mathcal{H}_j$.
Such a hyperedge is counted by $X_{i_1,i_2}^{(2,1)}$; therefore, by the same argument used for
$d_g^{(1,1)}(\mathcal{H}_i)$, one can obtain
\[
d_g^{(2,1)}(\mathcal{H}_{i_1},\mathcal{H}_{i_2})
\le X_{i_1,i_2}^{(2,1)}.
\]
Applying a Chernoff bound \ref{eq:chernoff2}, we obtain
\begin{align*}
    \mathbb{P}\left(d_g^{(2,1)}(\mathcal{H}_{i_1},\mathcal{H}_{i_2}) > 18\overline{m}^{1/3}\right) \leq \mathbb{P}\left(X_{i_1,i_2}^{(2,1)} > 18\overline{m}^{1/3}\right)
    \leq \left(\dfrac{e\mu_{X_{i_1,i_2}^{(2,1)}}}{18\overline{m}^{1/3}}\right)^{18\overline{m}^{1/3}}
    \leq \left(\dfrac{13e}{36}\right)^{18\overline{m}^{1/3}}
\end{align*}
Finally, by a union bound over all $\{i_1,i_2\} \subseteq [g]$, this yields:
\begin{align*}
    \mathbb{P}\left(d_g^{(2,1)} > 18\overline{m}^{\frac{1}{3}}\right) \leq \sum_{\{i_1,i_2\} \subseteq [g]}{\mathbb{P}\left(d_g^{(2,1)}(\mathcal{H}_{i_1},\mathcal{H}_{i_2}) > 18\overline{m}^{\textstyle\frac{1}{3}}\right)} 
    \leq \dfrac{g^2}{2} \cdot \left(\dfrac{13e}{36}\right)^{18\overline{m}^{1/3}} 
\end{align*}

We notice that $\dfrac{g^2}{2} \cdot \left(\dfrac{13e}{36}\right)^{18\overline{m}^{1/3}}  \to 0$ much faster than $\dfrac{1}{\log_3{n}}$ when $n \to +\infty$, thus it is still $o(1)$ after we take a union bound over all levels $\ell$ (with at most $\log_3{n}$ levels). It follows that with probability $1 - o(1)$, \eqref{eq:d21} holds for all $\ell$.

Combining all the above results, we can come to the conclusion that all conditions in the definition of $\mathcal{T}(\epsilon_n)$ hold with probability at least $1 - o(1)$ as $n \to +\infty$, which completes the proof of lemma \ref{lem:typical}.

\section{Analysis of level hypergraph}
\label{section:analysis_of_level_hypergraph}
\subsection{Proof of lemma \ref{lem:expectation}}
It suffices to show that, for any non-defective triple $(\Hc_u,\Hc_v,\Hc_w)$,
\[
\PP[\Ec_{uvw}] \leq \frac{1}{C}
\]
for some constant $C \geq 336$. According to the testing procedure, a given non-defective triple $(\Hc_u,\Hc_v, \Hc_w)$ fails to be identified correctly in a given sequence of $C_1\mb^{1/3}$ tests if either:
\begin{itemize}
  \item $\Hc_u$, $\Hc_v$ and $\Hc_w$ are not assigned to the same test, which occurs with probability $1 - \frac{1}{C_1^2 \mb^{2/3}}$;
  \item $\Hc_u$, $\Hc_v$ and $\Hc_w$ are assigned to the same test, but the outcome is positive because of other hyperedges.
\end{itemize}
Hence, the probability that $(\Hc_u,\Hc_v, \Hc_w)$ is missed in all $C_2\mb^{2/3}$ rounds is
\begin{equation}
\label{eq:prob_non_defective_pair}
    \PP[\Ec_{uvw}] = \left(1 - \frac{1}{C_1^2\mb^{2/3}} + \frac{1}{C_1^2\mb^{2/3}} \PP[Y=1 \mid h(u)=h(v) = h(w)] \right)^{C_2\mb^{2/3}}.
\end{equation}

Let $t := h(u)=h(v)=h(w)$. Since $(\Hc_u,\Hc_v,\Hc_w)$ is assumed to be non-defective, a false positive in test $t$ can only arise from some other defective configuration contained in the same test. To capture all such possibilities, define the following events.

\begin{itemize}
    \item $\Ac^{(1,1)}$: there exist $x \in \{u,v,w\}$ and $z \notin \{u,v,w\}$ such that $h(z) = t$ and $(\Hc_x,\Hc_z)$ are defective but $\Hc_z$ is non-defective.

    \item $\Ac^{(1,2)}$: there exist $x \in \{u,v,w\}$ and distinct $z_1,z_2 \notin \{u,v,w\}$ such that $h(z_1)=h(z_2)=t$,
    $(\Hc_x,\Hc_{z_1},\Hc_{z_2})$ is defective but $(\Hc_x,\Hc_{z_1})$, $(\Hc_x,\Hc_{z_2})$ and $(\Hc_{z_1},\Hc_{z_2})$ are non-defective. 

    \item $\Ac^{(2,1)}$: there exist distinct $x,y \in \{u,v,w\}$ and $z \notin \{u,v,w\}$ such that $h(z)=t$, $(\Hc_x,\Hc_y,\Hc_z)$ \text{is defective}, but $(\Hc_x,\Hc_{z})$ and $(\Hc_{y},\Hc_{z})$ are non-defective. 

    \item $\Ac^{0}$: there exists a defective hyperedge in the induced block hypergraph formed by all blocks $\Hc_z$ satisfying $z \notin \{u,v,w\}$ and $h(z)=t$.
\end{itemize}

Then every false positive outcome for the non-defective triple $(\Hc_u,\Hc_v,\Hc_w)$ is contained in the event
\[
\Ac^{(1,1)} \cup \Ac^{(1,2)} \cup \Ac^{(2,1)} \cup \Ac^{0}.
\]Then 
\begin{align*}
    \PP(Y=1 \mid h(u)=h(v) = h(w))
\le &\PP(\Ac^{(1,1)} \mid h(u)=h(v) = h(w)] + \PP[\Ac^{(1,2)} \mid h(u)=h(v)=h(w))\\
&+\PP(\Ac^{(2,1)} \mid h(u)=h(v)=h(w)] + \PP[\Ac^0 \mid h(u)=h(v)=h(w)) \, .
\end{align*}

From Lemma~\ref{lem:typical}, each non-defective node $u$ has at most $d_{g}^{(1,1)}$ blocks $z$ such that $(\Hc_u,\Hc_z)$ defective and at most $d_{g}^{(1,2)}$ pair blocks $(z_1,z_2)$ such that $(\Hc_u,\Hc_{z_1},\Hc_{z_2})$ defective, so
\[
\PP(\Ac^{(1,1)} \mid h(u)=h(v) = h(w)) \leq \frac{3d_{g}^{(1,1)}}{C_1 \mb^{1/3}} 
  \leq \frac{54\overline{m}^{1/3}}{C_1\mb^{1/3}} = \frac{54}{C_1},
\]
 
Simultaneously, we also have:

\[
\PP(\Ac^{(1,2)} \mid h(u)=h(v) = h(w)) \leq \frac{3d_{g}^{(1,2)}}{C_1^2 \mb^{2/3}} 
  \leq \frac{27\overline{m}^{2/3}}{C_1^2\mb^{2/3}} = \frac{27}{C_1^2},
\]
since $g \in [\mb^{1/3},n]$. Similarly, each pair $(u,v)$ has at most $d_g^{(2,1)}$ blocks $z$ such that $(\Hc_u,\Hc_v,\Hc_z)$ defective, it implies

\[
\PP(\Ac^{(2,1)} \mid h(u)=h(v) = h(w)) \leq \frac{3d_{g}^{(2,1)}}{C_1 \mb^{1/3}} 
  \leq \frac{54\overline{m}^{1/3}}{C_1\mb^{1/3}} = \frac{54}{C_1},
\]

Now consider a test among the $C_1 \mb^{1/3}$ tests. Let $\Bc_1$ be the event that some defective block belongs to the test, and $\Bc_2$ be the event that some defective pair $(\Hc_{z_1},\Hc_{z_2})$ is included in the test with both $\Hc_{z_1}$ and $\Hc_{z_2}$ being non-defective — i.e., the hyperedge appears only between $\Hc_{z_1}$ and $\Hc_{z_2}$. Let $\Bc_3$ be the event that some defective triple $(\Hc_{z_1},\Hc_{z_2},\Hc_{z_3})$ is included in the test with $\Hc_{z_i}$ (for all $i = 1,2,3$)  being non-defective and any pair $(\Hc_{z_i},\Hc_{z_j})$ being non-defective pair. Since $G$ has $M$ hyperedges, there are at most $M$ such triples. Moreover, from Eq.~\eqref{eq:definition-typical-set} used in Lemma~\ref{lem:typical}, there are at most $\PV^1_{\max}$ defective nodes and at most $\PV^2_{\max}$ defective pairs among $\{\Gc_1,\Gc_2,\dots,\Gc_g \}$. Thus,
\begin{align}
\mathbb{P}(\Ac^0 | h(u) = h(v) = h(w)) 
   &\leq \mathbb{P}(\Bc_1) + \mathbb{P}(\Bc_2) + \mathbb{P}(\Bc_3)\\
   &\leq \PV_g^1 \cdot \frac{1}{C_1 \mb^{1/3}} 
       + \PV_g^2 \cdot \frac{1}{C_1^2 \mb^{2/3}} + M \cdot \frac{1}{C_1^3 \mb}\\
   &\leq \frac{3}{C_1} + \frac{9}{C_1^2} + \dfrac{1+\varepsilon_n}{C_1^3} 
       \leq \frac{14}{C_1}\, , \label{eq:bound-Y}
\end{align}

The last inequality comes from the fact that for sufficiently large $n$, we have $\varepsilon_n \leq 1$. Combining the bounds into~\eqref{eq:prob_non_defective_pair} and recalling the choice $C_2 = C_1^3$, we obtain
\[
\PP(\Ec_{uvw}) \leq \left(1 - \frac{1}{C_1^2\mb^{2/3}} + \frac{149}{C_1^3\mb^{2/3}} \right)^{C_2\mb^{2/3}}
  \leq \exp\!\left(-\frac{C_2(C_1-149)}{C_1^3}\right)
  = \exp\!\big(-(C_1-149)\big).
\]
This establishes the claim, since this is at most $\exp(-6)< \frac{1}{336}$ when $C_1 \geq 155$.

\subsection{Proof of lemma \ref{lem:variance}}
	We first upper bound the covariance $\Cov[\rE_{uvw}, \rE_{u'v'w'}]$ for any non-defective triples $(\Hc_u, \Hc_v, \Hc_w)$ and $(\Hc_{u'}, \Hc_{v'},\Hc_w')$. We have
	\begin{equation}
            \label{eq:cova_formula}
	    \Cov[\rE_{uvw}, \rE_{u'v'w'}] = \PP[\Ec_{uvw} \cap \Ec_{u'v'w'}] - \PP[\Ec_{uvw}] \PP[\Ec_{u'v'w'}] \,.
	\end{equation}
Set $B  = C_1 \mb^{1/3}$, $R = C_2 \mb^{2/3}$, and $s = |\{u,v,w\} \cap \{u',v',w'\} |$. Consider a non-defective triple $(\Hc_u, \Hc_v,\Hc_w)$ at level~$\ell$. Let $\Dc_{uvw}$ be the event that the triple $(\Hc_u, \Hc_v,\Hc_w)$ is not identified in one round consisting of $B$ tests, and let $\barD_{uvw}$ denote its complement. We consider the following cases:
	
	\begin{enumerate}
		\item \textbf{Case 1: $s= 0$}. We have
		\begin{equation}\label{eq:lem_cova_prob_Duv_cap}
		    \PP(\Dc_{uvw} \cup \Dc_{u'v'w'}) = 1 - \PP(\barD_{uvw} \cap \barD_{u'v'w'}) \,.
		\end{equation}
		We then observe that
		\begin{equation}
		    \label{eq:lem_cova_Dc_uv_cap}
            \begin{split}
                \PP(\barD_{uvw} \cap \barD_{u'v'w'}) 
			&= \frac{B(B - 1)}{B^{6}} \PP \left( \barD_{uvw} \cap \barD_{u'v'w'} \,\middle|\, h(u) = h(v) = h(w), h(u') = h(v') = h(w'), h(u) \neq h(u') \right) \\
			&\quad + \frac{1}{B^5} \PP \left( \barD_{uvw} \cap \barD_{u'v'w'} \,\middle|\, h(u) = h(v)  = h(w)= h(u') = h(v') = h(w') \right) \\
			&= \frac{1}{B^4} \PP \left( \barD_{uvw} \cap \barD_{u'v'w'} \,\middle|\, h(u) = h(v) = h(w), h(u') = h(v') = h(w'), h(u) \neq h(u') \right) + \bigO\left( \frac{1}{B^5} \right) \,.
            \end{split}
		\end{equation}
		Set
        \begin{align*}
            \gamma &\coloneqq \PP\left( \barD_{uvw} \,\middle|\, h(u) = h(v) = h(w) \right) \, \\
            \gamma' &\coloneqq \PP\left( \barD_{u'v'w'} \,\middle|\, h(u') = h(v') = h(w') \right) \, \\
            \beta &\coloneqq \PP\left( \barD_{uvw} \cap \barD_{u'v'w'} \,\middle|\, h(u) = h(v) = h(w), h(u') = h(v') = h(w'), h(u) \neq h(u') \right) \, .
        \end{align*}

		Let $A$ be the event that $h(u) = h(v) = h(w)$. By the law of total probability, we have
        \begin{align*}
            \PP(\barD_{uvw}) = \PP(A) \cdot \PP(\barD_{uvw} \mid A) + \PP(\bar{A}) \cdot \PP(\barD_{uvw} \mid \bar{A})
        \end{align*}
        If $\Hc_u, \Hc_v$ and $\Hc_w$ are not placed inside the same test, then the triple $(\Hc_u,\Hc_v,\Hc_w)$ can't be detected, it follows that $\PP(\barD_{uvw} \mid \bar{A}) = 0$. Consequently 
        
        \begin{align}
            \PP(\barD_{uvw}) = \PP(A) \cdot \PP(\barD_{uvw} \mid A) = \dfrac{B}{B^3} \cdot \gamma = \dfrac{\gamma}{B^2} \Rightarrow \PP[\Dc_{uvw}] = 1 - \frac{\gamma}{B^2}
            \label{eq:lem_cova_D_uv}
        \end{align}
        Because each round is executed independently, we have
        \begin{align}
            \PP(\Ec_{uvw}) = \left(1 - \dfrac{\gamma}{B^2} \right)^{R}
            \label{eq:lem_cova_E_uv}
        \end{align}
        Similarly, we can also establish that 
        \begin{align}
            \PP[\Dc_{u'v'w'}] = 1 - \dfrac{\gamma'}{B^2}, \quad \PP(\Ec_{u'v'w'}) = \left(1 - \dfrac{\gamma'}{B^2} \right)^{R}
            \label{eq:lem_cova_u'v'w'}
        \end{align}
        From Eq.~\eqref{eq:lem_cova_Dc_uv_cap}, one has  
		\begin{equation} \label{eq:lem_cova_D_uv_cup}
		    \PP(\Dc_{uvw} \cup \Dc_{u'v'w'}) = 1 - \frac{\beta}{B^4} + \bigO\left( \frac{1}{B^5} \right) \,.
		\end{equation}
		Moreover, since $\PP(\Dc_{uvw} \cap \Dc_{u'v'w'}) = \PP(\Dc_{uvw}) + \PP(\Dc_{u'v'w'}) - \PP(\Dc_{uvw} \cup \Dc_{u'v'w'})$. From Eqs.~\eqref{eq:lem_cova_D_uv}, \eqref{eq:lem_cova_u'v'w'} and \eqref{eq:lem_cova_D_uv_cup}, we have
		\begin{equation} \label{eq:lem_cova_Duv_cap_Du'v'}
		    \begin{split}
		        \PP(\Dc_{uvw} \cap \Dc_{u'v'w'}) 
			& =  \left(1 - \frac{\gamma}{B^2} \right) + \left(1 - \frac{\gamma'}{B^2} \right) - 1 + \frac{\beta}{B^4} + \bigO\left( \frac{1}{B^5} \right) \\
            & = 1 - \frac{\gamma + \gamma'}{B^2}  + \frac{\beta}{B^4} + \bigO\left( \frac{1}{B^5} \right)  \,.
		    \end{split}
		\end{equation}
		Since each round is independent, it follows that
		\begin{align}
			\PP(\Ec_{uvw} \cap \Ec_{u'v'w'}) = \left( 1 - \frac{\gamma + \gamma'}{B^2}  + \frac{\beta}{B^4} + \bigO\left( \frac{1}{B^5} \right) \right)^R \,,
            \label{eq:lem_cova_Euv_cap_Eu'v'}
		\end{align}
		 Substituting \eqref{eq:lem_cova_Euv_cap_Eu'v'} and \eqref{eq:lem_cova_E_uv} into~\eqref{eq:cova_formula} gives 
		\begin{align*}
			\left|\Cov[\rE_{uvw}, \rE_{u'v'w'}]\right| 
			&= \left| \left( 1 - \frac{\gamma + \gamma'}{B^2}  + \frac{\beta}{B^4} + \bigO\left( \frac{1}{B^5} \right) \right)^R - \left( 1 - \frac{\gamma}{B^2} \right)^{R} \left( 1 - \frac{\gamma'}{B^2} \right)^{R} \right| \\
            & = \left| \left(1 - \frac{\gamma + \gamma'}{B^2}  + \frac{\beta}{B^4} + \bigO\left( \frac{1}{B^5} \right) \right)^R - \left(1 - \frac{\gamma + \gamma'}{B^2} + \dfrac{\gamma\gamma'}{B^4}\right)^{R}  \right| \\
			& \le R\left|1 - \frac{\gamma + \gamma'}{B^2}  + \frac{\beta}{B^4} + \bigO\left( \frac{1}{B^5} \right) - \left(1 - \frac{\gamma + \gamma'}{B^2} + \dfrac{\gamma\gamma'}{B^4} \right) \right| \\
            & = R \left| \dfrac{\beta - \gamma\gamma'}{B^4} + \bigO\left( \frac{1}{B^5} \right)\right| = R \cdot \bigO\left( \frac{1}{B^4} \right) = \bigO\left( \frac{1}{B^2} \right)
		\end{align*}
        
		Here, we used the inequality $|x^R - y^R| \leq R|x-y|$ for $x,y \in [0;1]$, recalling that $B  = C_1 \mb^{1/3}$, $R = C_2 \mb^{2/3}$, $C_2 = C_1^3$ and $\gamma, \gamma',\beta \in [0,1]$.
		
		\item \textbf{Case 2: $s= 1$}. Without loss of generality, assume that $u = u'$. Similar to Eq.~\eqref{eq:lem_cova_Dc_uv_cap} in Case 1, we have
		\begin{align*}
			\PP(\barD_{uvw} \cap \barD_{uv'w'}) = \frac{1}{B^4} \PP\left( \barD_{uvw} \cap \barD_{uv'w'} \,\middle|\, h(u) = h(v) = h(w) = h(v') = h(w') \right) \,.
		\end{align*}
		Set
        \begin{align*}
            \gamma &\coloneqq \PP\left( \barD_{uvw} \,\middle|\, h(u) = h(v) = h(w) \right) \, \\
            \gamma' &\coloneqq \PP\left( \barD_{uv'w'} \,\middle|\, h(u) = h(v') = h(w') \right) \, \\
            \beta &\coloneqq \PP\left( \barD_{uvw} \cap \barD_{uv'w'} \,\middle|\, h(u) = h(v) = h(w) = h(v') = h(w') \right) \, .
        \end{align*}
        With a similar line of reasoning to case 1, we can see that 
        \begin{align*}
             \begin{cases}
                 \PP(\Dc_{uvw}) = 1 - \dfrac{\gamma}{B^2},\quad \PP(\Dc_{uv'w'}) = 1 - \dfrac{\gamma'}{B^2} \\ 
             \PP(\Dc_{uvw} \cup \Dc_{uv'w'}) = 1 - \PP(\barD_{uvw} \cap \barD_{uv'w'}) = 1 - \dfrac{\beta}{B^4}
             \end{cases}
        \end{align*}
        Then, similar to~\eqref{eq:lem_cova_Duv_cap_Du'v'}, we have  
		\begin{align*}
			\PP(\Dc_{uvw} \cap \Dc_{uv'w'}) 
			&= \PP(\Dc_{uvw}) + \PP(\Dc_{uv'w'}) - \PP(\Dc_{uvw} \cup \Dc_{uv'w'}) = 1 - \frac{\gamma + \gamma'}{B^2}  + \frac{\beta}{B^4}
		\end{align*}
		Since each round is executed independently, it follows that
		\begin{align*}
            \PP(\Ec_{uvw}) = \left(1 - \dfrac{\gamma}{B^2}\right)^R, \PP(\Ec_{uv'w'}) = \left(1 - \dfrac{\gamma'}{B^2}\right)^R, \quad  \PP(\Ec_{uvw} \cap \Ec_{uv'w'}) = \left(1 -  \dfrac{\gamma + \gamma'}{B^2} + \dfrac{\beta}{B^4} \right)^{R}
		\end{align*}
		Therefore, by the same argument as in Case 1, we have:
		\begin{align*}
			\left|\Cov[\rE_{uvw}, \rE_{uv'w'}]\right| 
			&= \left|\left(1 -  \dfrac{\gamma + \gamma'}{B^2} + \dfrac{\beta}{B^4} \right)^{R}  - \left(1 - \dfrac{\gamma}{B^2}\right)^R\left(1 - \dfrac{\gamma'}{B^2}\right)^R \right| \\ & = \left| \left(1 - \frac{\gamma + \gamma'}{B^2}  + \frac{\beta}{B^4} \right)^R - \left(1 - \frac{\gamma + \gamma'}{B^2} + \dfrac{\gamma\gamma'}{B^4}\right)^{R}  \right| \\
            & \le R\left|1 - \frac{\gamma + \gamma'}{B^2}  + \frac{\beta}{B^4} - \left(1 - \frac{\gamma + \gamma'}{B^2} + \dfrac{\gamma\gamma'}{B^4} \right)\right|
            \\& = R\left|\dfrac{\beta - \gamma\gamma'}{B^4}\right| = R \cdot \bigO\left( \frac{1}{B^4} \right) = \bigO\left( \frac{1}{B^2} \right) 
		\end{align*}

     \item \textbf{Case 3: $s= 2$}. Without loss of generality, assume that $u = u'$ and $v = v'$. Similar to Eq.~\eqref{eq:lem_cova_Dc_uv_cap} in Case 1 and Case 2, we have
		\begin{align*}
			\PP(\barD_{uvw} \cap \barD_{uvw'}) = \frac{1}{B^3} \PP\left( \barD_{uvw} \cap \barD_{uvw'} \,\middle|\, h(u) = h(v) = h(w) = h(w') \right) \,.
		\end{align*}
		Set
        \begin{align*}
            \gamma &\coloneqq \PP\left( \barD_{uvw} \,\middle|\, h(u) = h(v) = h(w) \right) \, \\
            \gamma' &\coloneqq \PP\left( \barD_{uvw'} \,\middle|\, h(u) = h(v) = h(w') \right) \, \\
            \beta &\coloneqq \PP\left( \barD_{uvw} \cap \barD_{uvw'} \,\middle|\, h(u) = h(v) = h(w) = h(w') \right) \, .
        \end{align*}
        Applying the same argument as in Case 1, we obtain 
        \begin{align*}
            \PP(\Dc_{uvw}) = 1 - \dfrac{\gamma}{B^2},\quad \PP(\Dc_{uvw'}) = 1 - \dfrac{\gamma'}{B^2},\quad\PP(\Dc_{uvw} \cup \Dc_{uv'w'}) = 1 - \PP(\barD_{uvw} \cap \barD_{uv'w'}) = 1 - \dfrac{\beta}{B^3}
        \end{align*}
        
        Then, similar to~\eqref{eq:lem_cova_Duv_cap_Du'v'}, we have  
		\begin{align*}
			\PP(\Dc_{uvw} \cap \Dc_{uvw'}) 
			= \PP(\Dc_{uvw}) + \PP(\Dc_{uvw'}) - \PP(\Dc_{uvw} \cup \Dc_{uvw'})= 1 - \frac{\gamma + \gamma'}{B^2} + \frac{\beta}{B^3}  \,,
		\end{align*}
		Since the rounds are independent, it follows that
		\begin{align*}
			\PP(\Ec_{uvw}) = \left(1 - \dfrac{\gamma}{B^2}\right)^R,\quad \PP(\Ec_{uvw'}) = \left(1 - \dfrac{\gamma'}{B^2}\right)^R, \quad  \PP(\Ec_{uvw} \cap \Ec_{uvw'}) = \left( 1 - \frac{\gamma + \gamma'}{B^2} + \frac{\beta}{B^3} \right)^{R} \,.
		\end{align*}
		Proceeding as in Case 1, we obtain
		\begin{align*}
			\left|\Cov[\rE_{uvw}, \rE_{uvw'}]\right| 
			&= \left|\left(1 - \frac{\gamma + \gamma'}{B^2} + \frac{\beta}{B^3}\right)^R  - \left(1 - \frac{\gamma}{B^2} \right)^{R}\left(1 - \frac{\gamma'}{B^2} \right)^{R}\right| \\
            & =\left|\left(1 - \frac{\gamma + \gamma'}{B^2} + \frac{\beta}{B^3}\right)^R  - \left(1 - \frac{\gamma + \gamma'}{B^2} + \dfrac{\gamma\gamma'}{B^4}\right)^{R}\right| \\
            & \leq R \left|1 - \frac{\gamma + \gamma'}{B^2} + \frac{\beta}{B^3} - \left(1 - \frac{\gamma + \gamma'}{B^2} + \dfrac{\gamma\gamma'}{B^4}\right)\right| \\
            & = R \left|\dfrac{\beta}{B^3} - \dfrac{\gamma\gamma'}{B^4}\right| = R \cdot \bigO\left( \frac{1}{B^3} \right) = \bigO\left( \frac{1}{B} \right)
		\end{align*}
	\end{enumerate}
We note that
\[
\Var_\ell(\rE_{uvw}) \le \mathbb{E}(\rE_{uvw}) = \PP(\Ec_{uvw}) \le
\exp(-(C_1-149)).
\]
Hence,
\[
\sum_{u,v,w}\Var_\ell(\rE_{uvw})
=
\bigO(\rEm).
\]
From the above analysis, we have
\[
\left|\Cov(\rE_{uvw},\rE_{u'v'w'})\right|
\leq
\begin{cases}
\bigO(B^{-2}), & s=0,\\
\bigO(B^{-2}), & s=1,\\
\bigO(B^{-1}), & s=2.
\end{cases}
\]
Since each of the three overlap cases contains at most $\bigO(\rEm^2)$ pairs of non-defective triples, summing the above covariance bounds yields
\[
\sum_{u,v,w,u',v',w'}
\Cov\!\left(\rE_{uvw},\rE_{u'v'w'}\right)
=
\bigO\!\left(\frac{\rEm^2}{B}\right),
\]
since $B^{-2}\le B^{-1}$. Consequently,
\[
\Var_\ell\!\left(\sum_{u,v,w}\rE_{uvw}\right)
=
\bigO(\rEm)
+
\bigO\!\left(\frac{\rEm^2}{B}\right)
=
\bigO\!\left(\frac{\rEm^2}{B}\right),
\]
where the last equality follows from $\rEm=\Omega(\mb)$; see~\eqref{eq:definition-typical-set}. This completes the proof. 

\end{document}